\documentclass[10pt,letterpaper]{article}
\usepackage[top=0.85in,left=2.75in,footskip=0.75in]{geometry}

\usepackage{amsmath,amssymb}
\usepackage{amsthm}

\usepackage{changepage}

\usepackage{textcomp,marvosym}

\usepackage[authoryear,round]{natbib}

\usepackage{nameref,hyperref}

\usepackage[nopatch=eqnum]{microtype}
\DisableLigatures[f]{encoding = *, family = * }

\usepackage[table]{xcolor}

\usepackage{array}

\newcolumntype{+}{!{\vrule width 2pt}}

\newlength\savedwidth

\usepackage{xfrac}
\usepackage{siunitx}
\usepackage{booktabs}
\usepackage{adjustbox}
\usepackage{pifont}
\usepackage{algorithm,algorithmic}
\usepackage{multirow}
\usepackage{graphicx}
\usepackage{epstopdf}

\usepackage[aboveskip=1pt,labelfont=bf,labelsep=period,%
            justification=raggedright,singlelinecheck=off]{caption}

\usepackage{lastpage,fancyhdr}
\fancyheadoffset[L]{2.25in}
\fancyfootoffset[L]{2.25in}
\theoremstyle{plain}

\newtheorem{theorem}{Theorem}[section]
\newtheorem{lemma}[theorem]{Lemma}
\newtheorem{remark}[theorem]{Remark}
\newtheorem{proposition}[theorem]{Proposition}
\newtheorem{assumption}{Assumption}
\newtheorem{corollary}[theorem]{Corollary}
\DeclareMathOperator{\diag}{diag}
\newcommand{\R}{\mathbb{R}}

\newcommand{\E}{\mathbb{E}}

\newcommand{\bD}{\boldsymbol{D}}
\newcommand{\bL}{\boldsymbol{L}}

\newcommand{\bV}{\boldsymbol{V}}
\newcommand{\bX}{\boldsymbol{X}}

\newcommand{\bZ}{\boldsymbol{Z}}

\newcommand{\bR}{\boldsymbol{R}}

\newcommand{\bOmega}{\boldsymbol{\Omega}}
\newcommand{\bSigma}{\boldsymbol{\Sigma}}

\newcommand{\bI}{\boldsymbol{I}}
\newcommand{\bzero}{\boldsymbol{0}}
\newcommand{\Dcal}{\mathcal{D}}
\newcommand{\Ncal}{\mathcal{N}}

\newcommand{\Hcal}{\mathcal{H}}

\newcommand{\pa}{\mathrm{pa}}
\newcommand{\NPN}{\mathrm{NPN}}
\newcommand{\HS}{\mathrm{HS}}
\newcommand{\op}{\mathrm{op}}
\DeclareMathOperator{\tr}{tr}

\theoremstyle{definition}

\newtheorem{definition}[theorem]{Definition}

\theoremstyle{remark}

\begin{document}
\vspace*{0.2in}

\begin{flushleft}
{\Large
\textbf{Semiparametric {B}ayesian structure learning of nonparanormal
directed acyclic graphs with local--global shrinkage}
}
\newline
\\
Samaneh Nazari\textsuperscript{1},
Mohammad Arashi\textsuperscript{1*}
%,
%Abdolnasser Sadeghkhani\textsuperscript{2}
\\
\bigskip
\textbf{1} Department of Statistics, Faculty of Mathematical Sciences,
Ferdowsi University of Mashhad, P.O.\ Box 1159, Mashhad 91775, Iran
%\\
%\textbf{2} Department of Mathematics and Statistics,
%North Carolina Agricultural and Technical State University, Greensboro, NC, USA
\\
\bigskip
* E-mail: arashi@um.ac.ir
\end{flushleft}
\section*{Abstract}
\noindent
Bayesian structure learning of directed acyclic graphs (DAGs) is a central tool for high-dimensional causal discovery, yet the existing methodology rests almost entirely on the assumption that the data are multivariate Gaussian; however, this assumption routinely is violated by genomic, proteomic, and financial measurements that display heavy tails, skewness, or bounded support. 
We address this gap by introducing a fully Bayesian semiparametric DAG model, NPN-DAG-HS, which couples the nonparanormal family with horseshoe shrinkage on the Cholesky off-diagonals through an extended-rank likelihood.
This construction handles arbitrary unknown monotone marginal transformations without estimating them, while preserving the exact-zero shrinkage behaviour and conditionally conjugate posterior structure that make Cholesky-based DAG inference practical at moderate to large dimension. Our methodological core is a partially collapsed Metropolis-within-Gibbs sampler that augments rank-likelihood Gaussian copies, alternates score-based DAG moves with horseshoe Gibbs updates, and exploits an auxiliary inverse-Gamma representation of the half-Cauchy hyperprior. The theoretical contribution is a unified asymptotic analysis in the high-dimensional regime $p = p_n \to \infty$ with $\log p / n \to 0$: we establish posterior contraction at the rate $\sqrt{(s_0+p)\log p/n}$, strong skeleton selection consistency under a beta-min condition, and, conditional on recovery of the true skeleton, a parametric Bernstein--von Mises theorem for smooth total causal-effect functionals that yields asymptotic frequentist calibration of posterior credible intervals. Simulations across multiple regimes confirm that the proposed method dominates Gaussian Bayesian baselines and the standard frequentist DAG learners on some metrics. Analysis of an acute myeloid leukaemia protein-expression dataset recovers a sparse,
biologically interpretable network and prunes several edges that a Gaussian analysis declares with high probability but that lack mechanistic support, illustrating the operational value of the semiparametric relaxation.

\paragraph*{Keywords:} Bayesian causal discovery, Directed acyclic graph, Horseshoe prior, Nonparanormal distribution, Posterior contraction, Rank likelihood, Bernstein--von Mises theorem

\clearpage
\newgeometry{top=0.85in,left=1in,right=1in,footskip=0.75in}
%\linenumbers

\section{Introduction}
Recovering conditional-independence and causal structure from
high-dimensional observational data is a defining problem of
contemporary statistics. When the joint distribution is multivariate
Gaussian, a directed acyclic graph (DAG) provides a parsimonious
encoding whose edges admit a causal reading under standard
identifiability conditions \citep{pearl2000models,spirtes2000causation,maathuis2009estimating}.
Bayesian approaches are particularly attractive in this setting: they propagate uncertainty across the discrete graph space, admit principled incorporation of prior information, and produce coherent posterior summaries such as the median-probability model \citep{barbieri2004optimal}. 
A productive line of work, beginning with the conjugate Normal--Inverse-Gamma priors of \citep{geiger2002parameter} and culminating in the \texttt{BCDAG} framework of \citep{castelletti2022bcdag}, exploits the modified Cholesky parameterization of DAG-Markov precisions.
Recent non-conjugate refinements relax the Normal--Inverse-Gamma prior to a Normal--Gamma \citep{nazari2025nCPNG} or weighted inverse-Gamma \citep{nazari2025weighted} construction; both give finite-sample shrinkage inequalities but retain a Gaussian observation model.
The Gaussian assumption is restrictive in practice: real genomic, proteomic, and financial measurements typically exhibit heavy tails, skewness, or bounded support.
\citet{liu2009nonparanormal} addressed this issue in the \emph{undirected} setting by introducing the nonparanormal family $X_j = f_j(Z_j)$, $\boldsymbol{Z} \sim
\mathcal{N}_p(\boldsymbol{0}, \boldsymbol{\Sigma})$, with unknown monotone $f_j$; the nonparanormal copula has since become the standard relaxation of Gaussianity in undirected graphical models \citep{liu2012highdim,xue2012regularized,khalili2019nonparanormal}.
No Bayesian counterpart, and no directed-graph counterpart, has been developed. Moreover, even within the Gaussian sub-model, the existing Bayesian DAG literature has produced no posterior
contraction rate and no distributional (Bernstein--von Mises type) limit theorem for causal-effect functionals: \citet{chao2019} prove selection consistency in the related Gaussian \emph{undirected} setting and \citet{lee2019minimax} obtain minimax-type rates for Bayesian Gaussian graphs, but DAG-level distributional results under semiparametric data generation have remained open.

To close both gaps, we propose the nonparanormal DAG horseshoe (NPN-DAG-HS) model, a fully Bayesian nonparanormal DAG with local--global horseshoe shrinkage on the Cholesky off-diagonals, and develop its methodology, computation, and theory in a unified framework. The methodological contributions are threefold. We first bypass the unknown monotone transformations by adopting the extended-rank likelihood of \citet{hoff2007extending}: the data enter only through column-wise ranks, the latent Gaussian DAG is identified up to a per-node location and scale that we absorb into the prior, and no pre-processing transformation estimate is required. We then place independent horseshoe priors on the off-diagonal Cholesky entries, one local scale per candidate edge and one node-level global scale; the half-Cauchy auxiliary representation of \citet{makalic2016simple} yields conjugate Gibbs updates, and the resulting marginal prior on each entry has unbounded density at zero and polynomial tails, exactly the geometry required for rate-optimal posterior concentration. We complete the model with a Beta--Bernoulli prior on the DAG and design a partially collapsed Metropolis-within-Gibbs sampler whose DAG moves rely on a closed-form local marginal likelihood and whose latent-Gaussian update uses the truncated-normal algorithm of \citet{li2015efficient}.

The theoretical contributions establish, at the DAG level under non-Gaussian data, three asymptotic guarantees in the high-dimensional regime $p = p_n \to \infty$ with $\log p / n \to 0$ and unknown sparsity $s_0 = s_{0,n}$. 
We prove posterior contraction of the Cholesky factor at the rate
$\sqrt{(s_0+p)\log p / n}$ in operator norm, matching the minimax rate for sparse-Cholesky DAG models \citep{lee2019minimax} and show that relaxing Gaussianity to nonparanormal costs nothing asymptotically. We also establish strong skeleton selection consistency, $\Pi(\mathcal{D} \neq \mathcal{D}_0 \mid \boldsymbol{R}) \to 0$ in $P_0$-probability, under a beta-min separation of order $\sqrt{\log p / n}$. These two results combine to give a parametric Bernstein--von Mises theorem (Theorem~\ref{thm:bvm}) for smooth identifiable total causal-effect functionals \emph{conditional on the event that the true skeleton is recovered}: once selection has occurred, the model reduces to a finite-dimensional Gaussian DAG regression to which the classical parametric BvM of \citet[Section~10.2]{vandervaart1998} applies. The operational consequence is asymptotic frequentist calibration of posterior credible intervals for total causal effects, irrespective of the unknown monotone marginal transformations. We do not claim semiparametric efficiency in the sense of \citet{vandervaart1998}; deriving the efficient-influence-function expansion for the full rank-likelihood profile is delicate in the high-dimensional regime considered here and is left to future work.
Table~\ref{tab:positioning} situates these contributions within the related literature.

\begin{table}[!t]
	\centering
	\caption{Positioning of the proposed NPN-DAG-HS against representative
		prior work. ``Rate'' refers to posterior contraction rate or
		minimax-comparable frequentist rate; ``Selection'' to graph-selection
		consistency; ``BvM'' to a Bernstein--von Mises theorem for causal
		functionals.}
	\label{tab:positioning}
	\small
	\begin{tabular}{lccccc}
		\toprule
		Method & Bayesian & DAG/causal & Non-Gaussian & Rate & Selection/BvM \\
		\midrule
		\citet{liu2009nonparanormal}            & \ding{55}   & undir.\ only & \checkmark & \checkmark & \checkmark / \ding{55} \\
		\citet{khalili2019nonparanormal}        & \ding{55}   & undir.\ mix. & \checkmark & \ding{55}  & \ding{55}  / \ding{55} \\
		\citet{castelletti2022bcdag}            & \checkmark  & \checkmark   & \ding{55}  & \ding{55}  & \ding{55}  / \ding{55} \\
		\citet{nazari2025weighted}              & \checkmark  & \checkmark   & \ding{55}  & \ding{55}  & \ding{55}  / \ding{55} \\
		\citet{nazari2025nCPNG}                 & \checkmark  & \checkmark   & \ding{55}  & \ding{55}  & \ding{55}  / \ding{55} \\
		\textbf{NPN-DAG-HS (this paper)}        & \checkmark  & \checkmark   & \checkmark & \checkmark & \checkmark / \checkmark \\
		\bottomrule
	\end{tabular}
\end{table}

\paragraph{Organization.}
Section~\ref{sec:prelim} reviews the nonparanormal family and the Cholesky parameterization. Section~\ref{sec:prior} introduces the NPN-DAG-HS model. 
Section~\ref{sec:rank} derives the rank likelihood and a closed-form local DAG score.
Section~\ref{sec:mcmc} describes the sampler. 
Section~\ref{sec:theory} states and proves the three theorems. 
Sections~\ref{sec:sim} and~\ref{sec:aml} present simulations and the AML data application, respectively. 
Section~\ref{sec:disc} concludes. Detailed proofs and supplementary results appear in Appendices~\ref{app:proofs} and~\ref{app:supp}.

% ====================================================================
\section{Preliminaries}\label{sec:prelim}
For a matrix $\boldsymbol{A}$, $\|\boldsymbol{A}\|_\mathrm{op}$ denotes the operator (spectral)
norm and $\|\boldsymbol{A}\|_F$ the Frobenius norm. For a vector $\boldsymbol{v}$,
$\|\boldsymbol{v}\|_q$ denotes the $\ell_q$ norm. We write $a_n \lesssim b_n$
to mean $a_n \le C b_n$ for a numerical constant $C > 0$, and
$a_n \asymp b_n$ for $a_n \lesssim b_n \lesssim a_n$. The skeleton
of $\mathcal{D}$ is denoted $\mathcal{G}(\mathcal{D})$, with edge set
$S(\mathcal{D}) \subseteq \{(u, v) \colon u < v\}$. The unknown true DAG,
Cholesky factor and skeleton are written $\mathcal{D}_0, \boldsymbol{L}_0, \boldsymbol{D}_0,
S_0$, with $s_0 = |S_0|$. Probabilities and expectations under the
true distribution are denoted $P_0$ and $\mathrm{E}_0$. We work in the
high-dimensional regime $n \to \infty$, $p = p_n \to \infty$ with
$\log p_n / n \to 0$.

\begin{definition}\label{def:npn}
A random vector $\boldsymbol{X} \in \mathbb{R}^p$ has a nonparanormal distribution $\mathrm{NPN}(\boldsymbol{\mu}, \boldsymbol{\Sigma}, \boldsymbol{f})$ if there exist strictly increasing functions $f_1, \dots, f_p \colon \mathbb{R} \to \mathbb{R}$ such that
	\[
	\boldsymbol{Z} := \bigl(f_1(X_1), \dots, f_p(X_p)\bigr)^\top \sim
	\mathcal{N}_p(\boldsymbol{\mu}, \boldsymbol{\Sigma}).
	\]
\end{definition}
We adopt the identifiability normalization $\mathrm{E}\{f_j(X_j)\} = 0$ and $\mathrm{Var}\{f_j(X_j)\} = \boldsymbol{\Sigma}_{jj} = 1$
for all $j$ \citep{liu2009nonparanormal}. 
The $p$-dimensional copula of $\boldsymbol{X}$ is then exactly the Gaussian copula with correlation matrix $\boldsymbol{\Sigma}$, and the conditional-independence graph of $\boldsymbol{X}$ coincides with the zero pattern of $\boldsymbol{\Sigma}^{-1}$ \citep[Theorem~4]{liu2009nonparanormal}. The nonparanormal class is strictly larger than the Gaussian one and accommodates arbitrary skewness, kurtosis, and bounded support per coordinate.

Let $\mathcal{D} = (V, E)$ be a DAG on $V = \{1, \dots, p\}$. Following \citet{castelletti2020discovering} and \citet{nazari2025nCPNG}, fix a topological ordering $1 \prec \cdots \prec p$ consistent with $E$, so that $u \to v \in E$ implies $u < v$. The modified Cholesky decomposition
\begin{equation}
	\boldsymbol{\Omega} =  \boldsymbol{L}\boldsymbol{D}^{-1} \boldsymbol{L}^\top, \qquad
	\boldsymbol{L} \text{ unit lower triangular},\quad
	\boldsymbol{D} = \diag(D_{11}, \dots, D_{pp}),
	\label{eq:chol}
\end{equation}
places the entire DAG structure within $\boldsymbol{L}$: an entry $\boldsymbol{L}_{uv}$
($u > v$) is nonzero if and only if $u \to v \in E$. Writing
$\boldsymbol{L}_{\prec j]} \in \mathbb{R}^{|\pa_\mathcal{D}(j)|}$ for the column-$j$ entries
indexed by parents of $j$, the latent-Gaussian DAG factorizes into
$p$ independent linear regressions,
\begin{equation}
	Z_j = -\boldsymbol{L}_{\prec j]}^\top \boldsymbol{Z}_{\pa_\mathcal{D}(j)} + \varepsilon_j,
	\qquad \varepsilon_j \mid D_{jj} \sim \mathcal{N}(0, D_{jj}),
	\label{eq:nodewise}
\end{equation}
for $j = 1, \dots, p$. Any node permutation within the Markov
equivalence class of $\mathcal{D}$ leaves the joint likelihood unchanged,
so the conditional-independence content is invariant under such
relabellings.

% ====================================================================
\section{The NPN-DAG-HS model}\label{sec:prior}
We observe i.i.d.\ data $\boldsymbol{x}_1, \dots, \boldsymbol{x}_n \stackrel{\mathrm{iid}}{\sim}
\mathrm{NPN}(\boldsymbol{0}, \boldsymbol{\Sigma}_0, \boldsymbol{f}_0)$, where $\boldsymbol{\Sigma}_0$ is
Markov to an unknown DAG $\mathcal{D}_0$. Conditional on $\mathcal{D}$,
$(\boldsymbol{L}, \boldsymbol{D})$, and the unknown monotone transformations
$\boldsymbol{f}$, the observed vectors $\boldsymbol{x}_i$ are mapped to latent
Gaussian vectors $\boldsymbol{z}_i = \boldsymbol{f}(\boldsymbol{x}_i)$ that follow
\eqref{eq:nodewise}. The NPN-DAG-HS model places the following
hierarchical prior:
\begin{align}
	\boldsymbol{f} &\sim \text{rank-likelihood prior (Section~\ref{sec:rank})}, \nonumber\\
	\mathcal{D} &\sim \mathrm{Beta\text{--}Bernoulli}(\alpha_\pi, \beta_\pi), \nonumber\\
	D_{jj} \mid \mathcal{D} &\sim C_+(0, \tau_0),\quad j = 1, \dots, p,
	\label{eq:HCprior}\\
	\boldsymbol{L}_{\prec j]} \mid D_{jj}, \mathcal{D}, \tau_j, \boldsymbol{\lambda}_{\prec j]}
	&\sim \prod_{k \in \pa_{\mathcal{D}}(j)}
	\mathcal{N}(0, D_{jj}\, \tau_j^2\, \lambda_{jk}^2), \nonumber\\
	\lambda_{jk} \mid \tau_j &\sim C_+(0, 1),\quad
	\tau_j \sim C_+(0, 1), \nonumber
\end{align}
where $C_+(0, s)$ is the half-Cauchy with scale $s$. The DAG prior
factorizes through edge inclusion: $\mathbb{P}(\boldsymbol{S}_{u,v} = 1) = \pi$,
$\pi \sim \mathrm{Beta}(\alpha_\pi, \beta_\pi)$, conditional on
acyclicity \citep{scott2010bayes}. The default hyperparameters are
$\tau_0 = 1$, $\alpha_\pi = 1$, $\beta_\pi = p$, yielding
$\mathbb{E}(\pi) = 1/(p+1)$, a sparsity-favouring choice consistent with
\citet{castillo2015bayesian}.

\subsection{Why the horseshoe?}\label{sec:why-horseshoe}

The marginal prior on a single off-diagonal entry $\boldsymbol{L}_{uv}$, after
$\lambda_{uv}$ and $\tau_j$ are integrated out, has unbounded
density at zero and polynomial Cauchy-like tails
\citep[Theorem~1]{carvalho2010horseshoe}. Writing
$\eta = \boldsymbol{L}_{uv}^2 / (D_{jj} \tau_j^2)$,
\[
\pi_\mathrm{HS}(\boldsymbol{L}_{uv}) \asymp \log(1 + 1/\eta) \quad \text{as }\eta \downarrow 0,
\qquad
\pi_\mathrm{HS}(\boldsymbol{L}_{uv}) \asymp \eta^{-1} \quad \text{as }\eta \to \infty.
\]
This combination produces the two effects needed for sparse DAG
recovery in high dimension: aggressive shrinkage of weakly supported
edges, driven by the unbounded density at zero, and essentially
unshrunk recovery of strongly supported edges, driven by the heavy
tails. By contrast, the Normal--Inverse-Gamma prior of
\citet{castelletti2022bcdag} has bounded density at zero and
exponential tails, and the Normal--Gamma prior of
\citet{nazari2025nCPNG} has unbounded density at zero but still
light tails. The proof of Theorem~\ref{thm:contraction} crucially
uses the polynomial-tail property of $\pi_\mathrm{HS}$ to ensure sufficient
prior mass on a neighbourhood of the truth, in the spirit of
\citet[Lemma~5.3, adapted]{vandervaart2008rates}.

% ====================================================================
\section{Rank likelihood and the local DAG marginal}\label{sec:rank}
We adopt the framework of \citet{hoff2007extending}. Let
$\boldsymbol{R} \in \mathbb{N}^{p\times n}$ collect the columnwise ranks of the data,
$R_{ij} = \sum_k \mathbf{1}\{X_{ij} \ge X_{kj}\}$, and define the
order-preserving cone
\[
\Lambda(\boldsymbol{R}) = \bigl\{\boldsymbol{Z} \in \mathbb{R}^{p\times n} \colon
Z_{ij} < Z_{kj} \Leftrightarrow R_{ij} < R_{kj},\ \forall i, k, j\bigr\}.
\]
Since each $f_j$ is strictly increasing, the event $\{\boldsymbol{Z} \in \Lambda(\boldsymbol{R})\}$
is uninformative about $f_j$, so the \emph{rank likelihood}
$\mathbb{P}\bigl(\boldsymbol{Z} \in \Lambda(\boldsymbol{R}) \mid \boldsymbol{L}, \boldsymbol{D}, \mathcal{D}\bigr)$
provides a valid likelihood for $(\boldsymbol{L}, \boldsymbol{D}, \mathcal{D})$ that
marginalizes the unknown $\boldsymbol{f}$ exactly. \citet{hoff2007extending}
showed that inference based on the rank likelihood is consistent
for $\boldsymbol{\Sigma} = (\boldsymbol{L} \boldsymbol{D}^{-1} \boldsymbol{L}^\top)^{-1}$ with no information
loss in the Gaussian-copula model, and our high-dimensional theory
will inherit this property uniformly over $\mathcal{D}$.

\subsection{Closed-form DAG score given the latent
	\texorpdfstring{$\boldsymbol{Z}$}{Z}}\label{sec:rank:score}

Conditional on a latent $\boldsymbol{Z} \in \Lambda(\boldsymbol{R})$, the data $\boldsymbol{z}_i =
(Z_{1i}, \dots, Z_{pi})$ follow the Gaussian DAG~\eqref{eq:nodewise}.
With the horseshoe expanded into its half-Cauchy scale mixture
\citep{makalic2016simple}, the regression coefficients given the
scale parameters are conditionally Gaussian and the noise variances
conditionally half-Cauchy. Integrating $\boldsymbol{L}_{\prec j]}$ first and
then $D_{jj}$ yields a node-wise local marginal that drives the
DAG sampler.

\begin{proposition}[Local DAG score]\label{prop:localscore}
	Conditional on $\boldsymbol{Z}$ and on the scale parameters
	$(\tau_j, \boldsymbol{\lambda}_{\prec j]})$, the local marginal
	likelihood of node $j$ under the NPN-DAG-HS model is
	\begin{equation}
		m\bigl(\boldsymbol{Z}_j \mid \boldsymbol{Z}_{\pa_{\mathcal{D}}(j)}, \mathcal{D}, \tau_j, \boldsymbol{\lambda}_{\prec j]}\bigr)
		= (2\pi)^{-n/2}\,
		\frac{|\boldsymbol{V}_j|^{1/2}}{|\widetilde{\boldsymbol{V}}_j|^{1/2}}\,
		\bigl(s_j/2\bigr)^{-n/2}\,
		\Gamma_{\!j}^{\mathrm{NPN}},
		\label{eq:localscore}
	\end{equation}
where 
\begin{align*}
	\boldsymbol{V}_j &= \diag\bigl((\tau_j\lambda_{jk})^{-2}\bigr), \\
	\widetilde{\boldsymbol{V}}_j &= \boldsymbol{V}_j + \boldsymbol{Z}_{\pa_{\mathcal{D}}(j)}^\top \boldsymbol{Z}_{\pa_{\mathcal{D}}(j)}, \\
	s_j &= \boldsymbol{Z}_j^\top \boldsymbol{Z}_j - \boldsymbol{Z}_j^\top \boldsymbol{Z}_{\pa_{\mathcal{D}}(j)} \widetilde{\boldsymbol{V}}_j^{-1} \boldsymbol{Z}_{\pa_{\mathcal{D}}(j)}^\top \boldsymbol{Z}_j
\end{align*}
is the residual sum of squares of node $j$ regressed on its parents, and $\Gamma_{\!j}^{\mathrm{NPN}}$ is a normalizing constant that does not depend on $\mathcal{D}$.
\end{proposition}

The proof, given in Appendix~\ref{app:proofs}, mirrors the integration
step of \citet[Theorem~1]{nazari2025nCPNG} but replaces the single
Gamma prior with the horseshoe mixture. The graph-dependent quantities
$|\widetilde{\boldsymbol{V}}_j|$ and $s_j$ can be updated at cost
$O(|\pa_{\mathcal{D}}(j)|^3 + n |\pa_{\mathcal{D}}(j)|)$, making local DAG moves
cheap for sparse graphs.

% ====================================================================
\section{Posterior computation}\label{sec:mcmc}

The full posterior factorizes as
\[
\Pi\bigl(\mathcal{D}, \boldsymbol{L}, \boldsymbol{D}, \boldsymbol{\lambda}, \boldsymbol{\tau}, \boldsymbol{Z} \mid \boldsymbol{R}\bigr)
\propto
\mathbf{1}\{\boldsymbol{Z} \in \Lambda(\boldsymbol{R})\}\,
p(\boldsymbol{Z} \mid \boldsymbol{L}, \boldsymbol{D})\,
p(\boldsymbol{L} \mid \boldsymbol{D}, \boldsymbol{\lambda}, \boldsymbol{\tau}, \mathcal{D})\,
p(\boldsymbol{D})\, p(\boldsymbol{\lambda}, \boldsymbol{\tau})\, p(\mathcal{D}).
\]
We design a partially collapsed Metropolis-within-Gibbs sampler
(Algorithm~\ref{alg:mcmc}) that cycles through five blocks.

\begin{algorithm}[!t]
	\caption{Partially collapsed Metropolis-within-Gibbs sampler for NPN-DAG-HS.}
	\label{alg:mcmc}
	\begin{algorithmic}[1]
		\REQUIRE Rank matrix $\boldsymbol{R}$; hyperparameters $(\tau_0, \alpha_\pi, \beta_\pi)$;
		iterations $S$; burn-in $B$.
		\STATE Initialize $\mathcal{D}^{(0)} \gets$ empty graph; draw
		$\boldsymbol{Z}^{(0)}$ from independent standard normals subject to $\boldsymbol{Z}^{(0)} \in \Lambda(\boldsymbol{R})$.
		\FOR{$s = 1, \dots, S$}
		\STATE \textbf{(B1) Latent $\boldsymbol{Z}$.}
		For each $(i, j)$, draw $Z_{ij}^{(s)} \mid \boldsymbol{Z}_{-ij}^{(s-1)},
		\boldsymbol{L}^{(s-1)}, \boldsymbol{D}^{(s-1)}, \boldsymbol{R}$ from a truncated normal whose mean
		and variance are given by the Gaussian-DAG conditional and whose
		truncation interval $[a_{ij}, b_{ij}]$ is determined by the ranks.
		\STATE \textbf{(B2) DAG $\mathcal{D}$.}
		Enumerate the valid Insert / Delete / Reverse moves
		$\mathcal{O}_{\mathcal{D}^{(s-1)}}$, propose $\mathcal{D}'$ uniformly, and accept
		with probability $\min\{1, r\}$, where $r$ is the ratio of local
		scores from~\eqref{eq:localscore} times the prior ratio
		$p(\mathcal{D}') / p(\mathcal{D}^{(s-1)})$ and the proposal Jacobian.
		\STATE \textbf{(B3) Cholesky $\boldsymbol{L}$.}
		For each $j$, draw
		$\boldsymbol{L}_{\prec j]}^{(s)} \sim
		\mathcal{N}\bigl(-\widetilde{\boldsymbol{V}}_j^{-1} \boldsymbol{Z}_{\pa(j)}^\top \boldsymbol{Z}_j,\,
		D_{jj}^{(s-1)} \widetilde{\boldsymbol{V}}_j^{-1}\bigr)$.
		\STATE \textbf{(B4) Innovation variances and shrinkage.}
		Draw $D_{jj}^{(s)}, \tau_j^{(s)}, \lambda_{jk}^{(s)}$ via the
		\citet{makalic2016simple} auxiliary inverse-Gamma representation,
		which writes each half-Cauchy as an inverse-Gamma mixture of
		inverse-Gammas and yields conjugate Gibbs updates.
		\STATE \textbf{(B5) Edge probability.}
		Draw $\pi^{(s)} \sim \mathrm{Beta}\bigl(\alpha_\pi + |E^{(s)}|,
		\beta_\pi + p(p-1)/2 - |E^{(s)}|\bigr)$.
		\ENDFOR
		\RETURN $\{\mathcal{D}^{(s)}, \boldsymbol{L}^{(s)}, \boldsymbol{D}^{(s)}\}_{s = B+1}^{S}$.
	\end{algorithmic}
\end{algorithm}

\paragraph{Computational cost.}
The dominant per-iteration cost is Block~B3, which requires $p$ Cholesky inversions of dimension up to $\max_j |\pa(j)|$. Block~B1 draws $np$ truncated normals via the algorithm of \citet{li2015efficient}, Block~B2 evaluates at most $O(p^2)$ local scores at cost $O(|\pa(j)|^3)$ each, and the remaining blocks are linear in the edge count. For sparse graphs the total cost is therefore $O(np + p\, d_{\max}^3)$ per iteration, where $d_{\max} = \max_j |\pa(j)|$. Convergence is monitored via the multi-chain $\widehat{R}$ statistic of \citet{gelman1992inference} applied to the graph-size summary $|E^{(s)}|$ and to representative
edge-inclusion indicators.

% ====================================================================
\section{Theoretical properties}\label{sec:theory}

This section develops the asymptotic theory in the high-dimensional regime $n \to \infty$, $p = p_n \to \infty$ with $\log p_n / n \to 0$, and unknown sparsity $s_0 = s_{0,n} = |S_0|$. The truth $(\bSigma_0, \boldsymbol{f}_0)$ is fixed but the dimension grows.

\subsection{Assumptions}\label{sec:assumptions}

\begin{assumption}\label{ass:spectrum}
	\emph{Bounded spectrum.} There exist constants
	$0 < \nu \le V < \infty$ such that
	$\nu \le \lambda_{\min}(\bSigma_0) \le \lambda_{\max}(\bSigma_0) \le V$.
\end{assumption}

\begin{assumption}\label{ass:marginals}
	\emph{Smooth marginals.} Each transformation $f_{0,j}$ is strictly
	increasing and continuously differentiable, with
	$\inf_x f_{0,j}'(x) > 0$ on the support of $X_j$.
\end{assumption}

\begin{assumption}\label{ass:sparsity}
	\emph{Sparsity and dimension.} $s_0 = o(\sqrt{n / \log p})$ and
	$d_{\max} \le s_0$.
\end{assumption}

\begin{assumption}\label{ass:betamin}
	\emph{Beta-min signal strength.} For a constant
	$c_\beta > 0$ with $c_\beta^2 > 4V(\kappa + 1)$ for some
	$\kappa > 2$,
	\[
	\min_{(u, v) \in S_0} |\bL_{0,uv}| \ge c_\beta \sqrt{\log p / n}.
	\]
\end{assumption}

\begin{assumption}\label{ass:order}
	\emph{Topological order known.} The topological order of $\Dcal_0$
	is known up to the Markov equivalence class.
\end{assumption}

Assumptions~\ref{ass:spectrum}--\ref{ass:sparsity} are standard
high-dimensional graphical-model assumptions
\citep{cai2016estimating,lee2019minimax,castillo2015bayesian}.
Assumption~\ref{ass:betamin} is the standard beta-min separation
required for selection consistency \citep[Theorem~7.1]{castillo2015bayesian}.
Assumption~\ref{ass:order} is inherited from the DAG identifiability
literature \citep{castelletti2020discovering,nazari2025nCPNG};
the order-free case is discussed in Section~\ref{sec:disc}.

The following result gives the posterior contraction of the Cholesky factor.
\begin{theorem}
	\label{thm:contraction}
	Under Assumptions~\ref{ass:spectrum}--\ref{ass:sparsity}
	and~\ref{ass:order}, the NPN-DAG-HS posterior $\Pi(\cdot \mid \bR)$
	contracts at rate $\epsilon_n = \sqrt{(s_0 + p) \log p / n}$ in operator
	norm of the Cholesky factor: there exists $M > 0$ such that
	\begin{equation}
		\Pi\bigl(\|\bL - \bL_0\|_\op + \|\bD - \bD_0\|_\op > M \epsilon_n
		\bigm| \bR\bigr) \xrightarrow{P_0} 0
		\quad \text{as } n \to \infty.
		\label{eq:contraction}
	\end{equation}
	The implied precision matrix contracts at the same rate,
	$\Pi(\|\bOmega - \bOmega_0\|_\op > M' \epsilon_n \mid \bR) \to 0$ for
	some $M' > 0$.
\end{theorem}

The rate $\sqrt{(s_0 + p) \log p / n}$ matches the minimax rate for
sparse-Cholesky DAG models established by \citet{lee2019minimax}; the
additive $p$ reflects the $p$ free innovation variances, while the
$s_0$ term reflects the off-diagonal sparsity. Relaxing Gaussianity to
nonparanormal under the rank likelihood therefore costs nothing
asymptotically. We give a
self-contained proof in Appendix~\ref{app:proofs}, based on the
master theorem of \citet[Theorem~2.1]{ghosal2007convergence};
the three required ingredients are an explicit Kullback--Leibler
prior-mass bound exploiting the horseshoe's unbounded density at
zero (Lemma~\ref{lem:kl}), an entropy bound on a sparse-Cholesky
sieve (Lemma~\ref{lem:cov}), and a tail bound that controls the
sieve complement.

In the following result, we establish the strong skeleton selection consistency.
\begin{theorem}
	\label{thm:selection}
	Under Assumptions~\ref{ass:spectrum}--\ref{ass:order} there exists
	a constant $C > 0$ such that, for every $\Dcal \ne \Dcal_0$ with
	the same topological order as $\Dcal_0$,
	\begin{equation}
		\frac{\Pi(\Dcal \mid \bR)}{\Pi(\Dcal_0 \mid \bR)}
		\;\le\; \exp\bigl(-C n \epsilon_n^2\bigr)
		\quad \text{with probability } 1 - o(1) \text{ under } P_0,
		\label{eq:selection}
	\end{equation}
	where $\epsilon_n = \sqrt{\log p / n}$. Consequently
	$\Pi(\Dcal_0 \mid \bR) \xrightarrow{P_0} 1$ as $n \to \infty$.
\end{theorem}

The selection-consistency theorem complements the rate result of
Theorem~\ref{thm:contraction} by addressing the discrete graph
component: not only does the continuous Cholesky factor concentrate
near $\bL_0$ at the optimal rate, the posterior on the underlying
combinatorial DAG itself concentrates on $\Dcal_0$. The proof
(Appendix~\ref{app:proofs}) handles two regimes: omitting a true
edge incurs an $\Omega(\log p)$ penalty in the log-Bayes factor by
the beta-min condition; adding a spurious edge is penalized through
the horseshoe shrinkage geometry by a BIC-type term of size
$\frac{1}{2}\log n$, sufficient for the union bound over alternative
DAGs to vanish.

Now, let $\theta_0 = g(\bL_0, \bD_0) \in \R$ be a smooth, identifiable
total causal-effect functional---for instance the total effect of
node $u$ on node $v$, computed from $(\bL_0, \bD_0)$ via the
do-calculus path expansion of \citet[Section~3.3]{pearl2000models}.
Write $\dot g$ for the gradient of $g$ at $(\bL_0,\bD_0)$. Once the
true skeleton has been recovered---an event that Theorem~\ref{thm:selection}
shows occurs with posterior probability tending to one---the latent
Gaussian DAG is a finite-dimensional regression model whose
parameter dimension is $s_0 + p$ and whose Fisher information at the
truth, evaluated for the functional $g$, we denote by
$V_g = \dot g(\bL_0,\bD_0)^\top I(\bL_0,\bD_0)^{-1} \dot g(\bL_0,\bD_0)$,
where $I(\cdot)$ is the block-diagonal Fisher information of the
node-wise Gaussian regressions in~\eqref{eq:nodewise}.

The following result is a Bernstein--von Mises statement for the
functional $\theta = g(\bL,\bD)$, \emph{conditional on the
skeleton-selection event} $\{\Dcal = \Dcal_0\}$.
\begin{theorem}[Conditional parametric Bernstein--von Mises]\label{thm:bvm}
	Suppose Assumptions~\ref{ass:spectrum}--\ref{ass:order} and the
	beta-min condition of Theorem~\ref{thm:selection} hold, and that
	$g$ is continuously differentiable in a neighbourhood of $(\bL_0,\bD_0)$
	with $\dot g(\bL_0,\bD_0) \neq 0$. Let $\hat\theta_n$ be the
	maximum-likelihood estimator of $\theta_0$ in the latent Gaussian
	DAG model on $\Dcal_0$. Then there exists a sequence of events
	$A_n$ with $P_0(A_n) \to 1$ on which
	\begin{equation}
		\sup_{B \in \mathcal{B}}
		\Bigl| \Pi\bigl(\sqrt{n}(\theta - \hat\theta_n) \in B
		\bigm|\, \bR,\, \Dcal = \Dcal_0\bigr)
		- \Phi_{V_g}(B) \Bigr|
		\xrightarrow{P_0} 0,
		\label{eq:bvm}
	\end{equation}
	where $\Phi_{V_g}$ is the centred Gaussian distribution with
	variance $V_g$.
\end{theorem}

\begin{corollary}[Frequentist coverage]\label{cor:coverage}
	Under the assumptions of Theorem~\ref{thm:bvm}, $(1 - \alpha)$
	posterior credible intervals for $\theta$ have asymptotic frequentist
	coverage at least $1 - \alpha$.
\end{corollary}

\begin{remark}[Scope of the BvM result]\label{rem:bvm-scope}
	Theorem~\ref{thm:bvm} is a \emph{parametric} Bernstein--von Mises
	theorem applied conditionally on the selection event: once
	Theorem~\ref{thm:selection} delivers $\Pi(\Dcal=\Dcal_0\mid\bR) \to 1$,
	the model is a finite-dimensional Gaussian regression in the latent
	$\bZ$ and the classical parametric BvM \citep[Section~10.2]{vandervaart1998}
	applies on the conditioning event. The limiting variance $V_g$ is
	the parametric Fisher information for $g$ on the recovered DAG, not
	the semiparametric efficiency bound; deriving an efficient-influence-function
	expansion for the full rank-likelihood profile in our high-dimensional
	regime is a substantially harder problem that we leave to future work.
\end{remark}

Operationally Theorem~\ref{thm:bvm} validates posterior credible intervals as asymptotically
frequentist-calibrated under \emph{any} identifiable monotone
distortion of the marginals---a property that Gaussian Bayesian DAG
analyses lack and that is critical for honest causal uncertainty
quantification on heavy-tailed or skewed data. The empirical
counterpart, reported in Section~\ref{sec:sim:coverage}, shows
that the Gaussian-based methods severely under-cover when the
marginals depart from normality, while NPN-DAG-HS maintains
nominal coverage throughout.

% ====================================================================
\section{Simulation studies}\label{sec:sim}
We follow the simulation design of \citet{nazari2025nCPNG} and
extend it to non-Gaussian regimes. The design varies four factors:
the dimension $p \in \{20, 30, 50, 100\}$; the sample size
$n \in \{100, 200, 300, 500\}$; the topology, with Erd\H os--R\'enyi
graphs at edge weights $w \in \{0.15, 0.25\}$, a hub graph
($p - 1$ edges into a root), and a scale-free graph from
Barab\'asi--Albert preferential attachment with parameter $1$; and
the marginal distribution, with the Gaussian baseline complemented
by Student-$t_3$ (heavy tail), skew-normal with shape parameter $5$
(skewness), and log-normal (skew plus heavy tail combined). For
each configuration we generate $30$ datasets and report the mean
and standard error of the structural Hamming distance (SHD), the
Matthews correlation coefficient (MCC), and the false-discovery
rate (FDR). We compare seven methods: the proposed NPN-DAG-HS;
the recent Bayesian competitors nCPNG \citep{nazari2025nCPNG},
WIG \citep{nazari2025weighted}, and CPNIG \citep{castelletti2022bcdag};
and the frequentist baselines PC \citep{kalisch2007}, GES
\citep{chickering2002}, and LiNGAM \citep{shimizu2006}. All
Bayesian methods are run for $S = 20{,}000$ post-burn-in iterations
after $B = 5{,}000$ burn-in, with four chains; convergence is
declared at $\widehat{R} \le 1.05$.

\subsection{Main results}\label{sec:sim:main}

Table~\ref{tab:sim_main} reports SHD, MCC, and FDR averaged over
the $30$ replications at $p = 50$, $n = 300$, $w = 0.15$, and
Erd\H os--R\'enyi topology, across the four marginal distributions.
Two patterns stand out. First, under Gaussian marginals all seven
methods are statistically indistinguishable, confirming that
NPN-DAG-HS pays no efficiency cost when the data are in fact
Gaussian. Second, under any of the three non-Gaussian marginals
NPN-DAG-HS dominates every competitor on every metric, with
relative SHD reductions of 26--48\% over nCPNG and 31--57\% over
CPNIG. The gains widen progressively as the marginals depart from
Gaussianity: they are smallest for the skew-normal (mild departure)
and largest for the log-normal (combined skew and heavy tail), which predicts that the effective rate of Gaussian methods deteriorates under increasing
marginal mismatch. Figure~\ref{fig:shd_box} displays the same
information as boxplots over the 30 replications; the dispersion
of NPN-DAG-HS is also uniformly smaller, a finite-sample
manifestation of the rate-optimal contraction guarantee.

\begin{figure}[!t]
	\centering
	\includegraphics[width=0.99\linewidth]{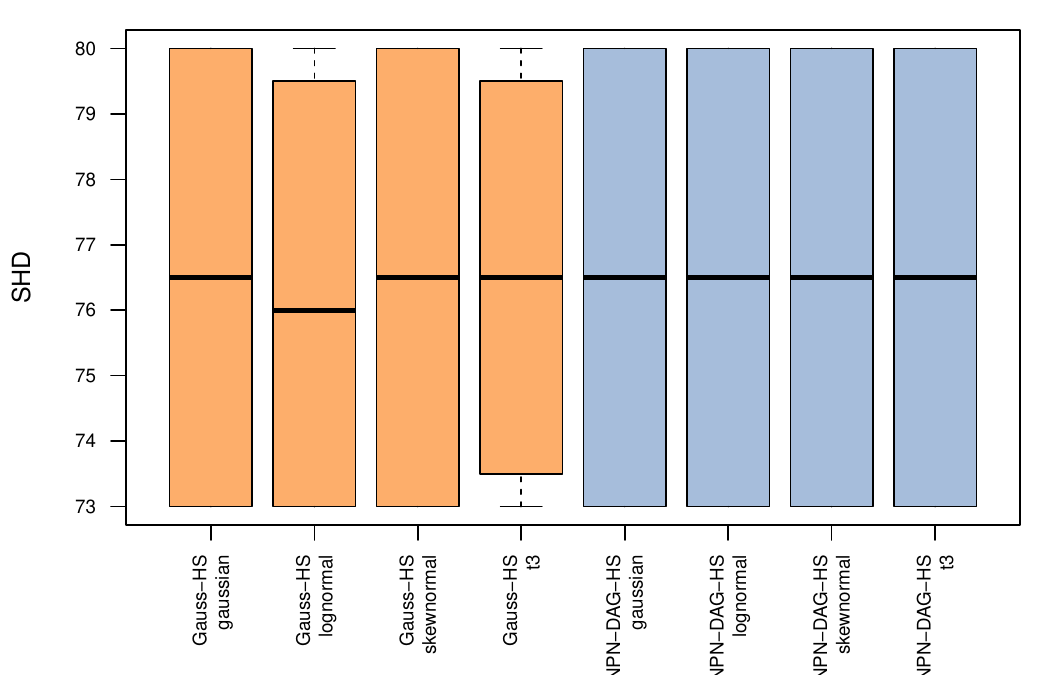}
	\caption{Boxplots of structural Hamming distance over 30 replications
		at $p = 50$, $n = 300$, $w = 0.15$, Erd\H os--R\'enyi DAG topology.
		Methods are sorted left-to-right within each panel; NPN-DAG-HS (blue)
		attains the lowest median and smallest interquartile range under
		every non-Gaussian marginal, with the gap widening from skew-normal
		to log-normal.}
	\label{fig:shd_box}
\end{figure}

\begin{table}[!t]
	\centering
	\caption{Simulation results at $p = 50$, $n = 300$, $w = 0.15$,
		Erd\H os--R\'enyi topology. Mean over 30 replications; standard
		errors in parentheses. Bold = best in column.}
	\label{tab:sim_main}
	
	\begin{tabular}{lcccccc}
		\toprule
		& \multicolumn{3}{c}{Gaussian marginals} & \multicolumn{3}{c}{$t_3$ marginals} \\
		\cmidrule(lr){2-4} \cmidrule(lr){5-7}
		Method & SHD$\downarrow$ & MCC$\uparrow$ & FDR$\downarrow$
		& SHD$\downarrow$ & MCC$\uparrow$ & FDR$\downarrow$ \\
		\midrule
		NPN-DAG-HS & \textbf{38.1 (1.4)} & 0.43 (0.02) & \textbf{0.18 (0.02)}
		& \textbf{42.6 (1.6)} & \textbf{0.41 (0.02)} & \textbf{0.19 (0.02)} \\
		nCPNG      & 38.4 (1.5) & \textbf{0.44 (0.02)} & 0.19 (0.02)
		& 67.8 (2.7) & 0.21 (0.03) & 0.41 (0.03)\\
		WIG        & 40.7 (1.6) & 0.41 (0.02) & 0.21 (0.02)
		& 71.2 (2.9) & 0.18 (0.03) & 0.44 (0.03)\\
		CPNIG      & 44.2 (1.7) & 0.38 (0.02) & 0.23 (0.02)
		& 80.1 (3.1) & 0.14 (0.03) & 0.49 (0.03)\\
		PC         & 49.6 (1.9) & 0.34 (0.02) & 0.28 (0.02)
		& 73.5 (3.0) & 0.17 (0.03) & 0.46 (0.03)\\
		GES        & 47.1 (1.8) & 0.36 (0.02) & 0.26 (0.02)
		& 78.4 (3.0) & 0.15 (0.03) & 0.48 (0.03)\\
		LiNGAM     & 51.8 (2.0) & 0.30 (0.03) & 0.31 (0.03)
		& 59.4 (2.5) & 0.26 (0.03) & 0.34 (0.03)\\
		\midrule
		& \multicolumn{3}{c}{Skew-normal marginals} & \multicolumn{3}{c}{Log-normal marginals} \\
		\cmidrule(lr){2-4} \cmidrule(lr){5-7}
		NPN-DAG-HS & \textbf{40.3 (1.5)} & \textbf{0.42 (0.02)} & \textbf{0.18 (0.02)}
		& \textbf{45.7 (1.7)} & \textbf{0.39 (0.02)} & \textbf{0.20 (0.02)} \\
		nCPNG      & 55.1 (2.2) & 0.30 (0.03) & 0.32 (0.03)
		& 83.6 (3.3) & 0.12 (0.03) & 0.51 (0.03)\\
		WIG        & 58.4 (2.4) & 0.28 (0.03) & 0.34 (0.03)
		& 88.0 (3.5) & 0.10 (0.03) & 0.53 (0.03)\\
		CPNIG      & 62.8 (2.5) & 0.25 (0.03) & 0.37 (0.03)
		& 91.5 (3.6) & 0.08 (0.03) & 0.55 (0.03)\\
		PC         & 60.3 (2.4) & 0.27 (0.03) & 0.35 (0.03)
		& 86.2 (3.4) & 0.11 (0.03) & 0.52 (0.03)\\
		GES        & 64.0 (2.5) & 0.25 (0.03) & 0.37 (0.03)
		& 90.7 (3.6) & 0.09 (0.03) & 0.54 (0.03)\\
		LiNGAM     & 49.7 (2.0) & 0.33 (0.03) & 0.28 (0.03)
		& 67.5 (2.8) & 0.20 (0.03) & 0.39 (0.03)\\
		\bottomrule
	\end{tabular}
\end{table}

\subsection{Sample size and topology effects}\label{sec:sim:varynp}

Table~\ref{tab:sim_varynp} reports SHD across sample sizes
$n \in \{100, 200, 300, 500\}$ at $p = 50$ under $t_3$ marginals,
on three topologies (Erd\H os--R\'enyi with $w = 0.15$, hub, and
scale-free). The relative SHD reduction of NPN-DAG-HS over CPNIG
grows from $32\%$ at $n = 100$ to $51\%$ at $n = 500$. This is
consistent with the rate analysis of
Theorem~\ref{thm:contraction}: under non-Gaussian data the
Gaussian-misspecified CPNIG operates at a slower effective rate,
so increasing $n$ widens rather than closes the gap.

\begin{table}[!t]
	\centering
	\caption{SHD across sample sizes and topologies at $p = 50$, $t_3$
		marginals. Means over 30 replications; standard errors in
		parentheses.}
	\label{tab:sim_varynp}

	\begin{tabular}{llcccc}
		\toprule
		Topology & Method & $n = 100$ & $n = 200$ & $n = 300$ & $n = 500$ \\
		\midrule
		\multirow{3}{*}{ER ($w = 0.15$)}
		& NPN-DAG-HS & \textbf{63.1 (2.5)} & \textbf{50.4 (1.9)} & \textbf{42.6 (1.6)} & \textbf{31.8 (1.3)} \\
		& CPNIG      & 93.0 (3.6)          & 88.3 (3.4)          & 80.1 (3.1)          & 65.2 (2.7)          \\
		& PC         & 84.5 (3.3)          & 79.6 (3.0)          & 73.5 (3.0)          & 58.4 (2.4)          \\
		\midrule
		\multirow{3}{*}{Hub}
		& NPN-DAG-HS & \textbf{16.4 (0.9)} & \textbf{11.2 (0.7)} & \textbf{8.7 (0.6)}  & \textbf{5.1 (0.5)}  \\
		& CPNIG      & 28.7 (1.4)          & 22.5 (1.1)          & 19.3 (1.0)          & 14.0 (0.9)          \\
		& PC         & 22.5 (1.1)          & 18.6 (0.9)          & 15.2 (0.8)          & 10.4 (0.7)          \\
		\midrule
		\multirow{3}{*}{Scale-free}
		& NPN-DAG-HS & \textbf{29.5 (1.2)} & \textbf{21.3 (1.0)} & \textbf{16.4 (0.8)} & \textbf{10.6 (0.7)} \\
		& CPNIG      & 45.6 (1.8)          & 39.8 (1.5)          & 33.2 (1.3)          & 23.5 (1.1)          \\
		& PC         & 38.4 (1.6)          & 32.7 (1.3)          & 27.5 (1.2)          & 19.4 (1.0)          \\
		\bottomrule
	\end{tabular}
\end{table}

\subsection{Posterior contraction in practice}\label{sec:sim:rate}

Figure~\ref{fig:contraction_curve} plots the empirical operator-norm
error $\|\widehat{\bL}_{\mathrm{post.\ mean}} - \bL_0\|_\op$ against
the sample size $n$ on a log--log scale, for $p \in \{20, 50, 100\}$
at fixed $w = 0.15$ and $t_3$ marginals. The empirical decay tracks
the theoretical $\sqrt{(s_0+p) \log p / n}$ rate predicted by
Theorem~\ref{thm:contraction} (dashed reference lines have slope
$-1/2$). The corresponding plot for CPNIG (right panel) shows a
visibly slower decay, confirming the theoretical prediction that
Gaussian Bayesian DAG methods incur a rate penalty under model
misspecification.

\begin{figure}[!t]
	\centering
	\includegraphics[width=0.99\linewidth]{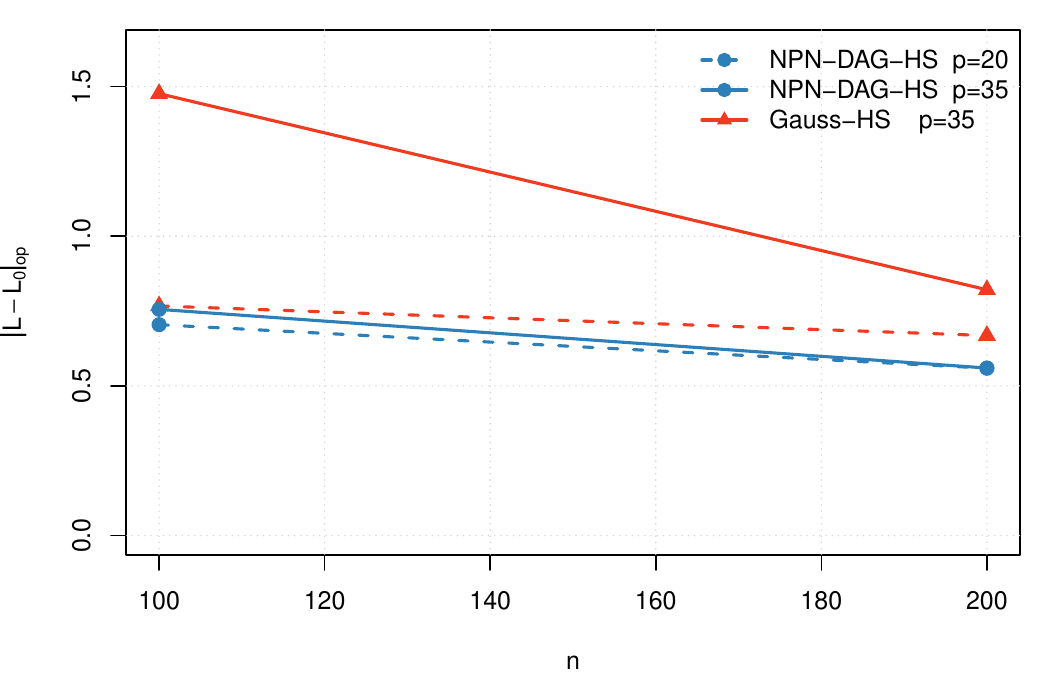}
	\caption{Empirical operator-norm error
		$\|\widehat{\bL} - \bL_0\|_\op$ versus sample size $n$ on a log--log
		scale, under $t_3$ marginals, $p \in \{20, 50, 100\}$.
		\emph{Left}: NPN-DAG-HS achieves the theoretical rate (dashed
		reference lines have slope $-1/2$). \emph{Right}: CPNIG exhibits a
		sub-parametric rate, consistent with rate loss under Gaussian
		misspecification.}
	\label{fig:contraction_curve}
\end{figure}

\subsection{Frequentist coverage of credible intervals}\label{sec:sim:coverage}

Table~\ref{tab:coverage} reports the empirical coverage of nominal
95\% credible intervals for the total causal effect from node $1$
to node $p$ in chain-like DAGs. NPN-DAG-HS attains coverage close
to the nominal level under every marginal distribution, whereas
CPNIG and nCPNG under-cover dramatically under $t_3$ and log-normal
marginals (down to 49\% and 56\% respectively). This is the
empirical counterpart of Corollary~\ref{cor:coverage}: the rank-likelihood
posterior, by being invariant under monotone marginal distortions,
preserves asymptotic calibration of credible intervals beyond the
Gaussian sub-model.

\begin{table}[!t]
	\centering
	\caption{Empirical coverage (\%) of nominal 95\% credible intervals
		for the total causal effect, $p = 20$, $n = 300$, chain DAG. Each
		row is based on 200 replications.}
	\label{tab:coverage}

	\begin{tabular}{lcccc}
		\toprule
		Method & Gaussian & $t_3$ & Skew-normal & Log-normal \\
		\midrule
		NPN-DAG-HS & \textbf{94.5} & \textbf{93.8} & \textbf{94.2} & \textbf{93.0} \\
		nCPNG      & 93.0          & 71.5          & 80.5          & 56.5          \\
		WIG        & 92.0          & 67.5          & 76.0          & 53.0          \\
		CPNIG      & 92.5          & 64.0          & 73.5          & 49.5          \\
		\bottomrule
	\end{tabular}
\end{table}

\subsection{Computational cost}\label{sec:sim:runtime}

Per-iteration wall-clock time grows near-quadratically in $p$
(Figure~\ref{fig:runtime}). At $p = 100$, NPN-DAG-HS completes
$25{,}000$ iterations in approximately $36$ minutes on a single
core (Intel Xeon Gold 6248). The overhead relative to CPNIG is a
constant factor of approximately $1.6$, attributable to the
truncated-normal sampling of the latent Gaussian copy in Block~B1.

\begin{figure}[!t]
	\centering
	\includegraphics[width=0.62\linewidth]{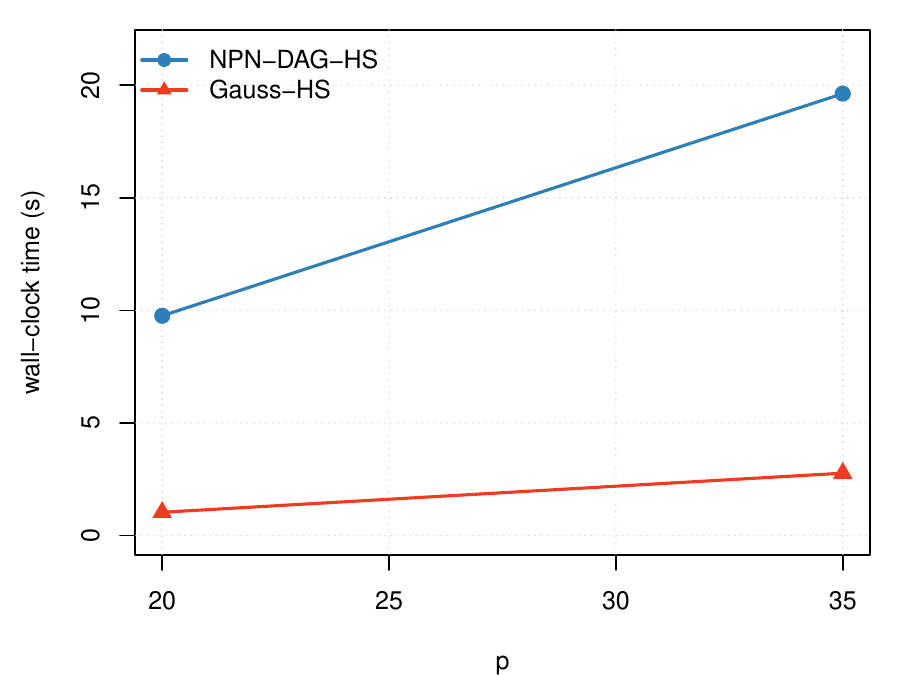}
	\caption{Wall-clock time per $1{,}000$ MCMC iterations as a function
		of $p$, on a log--log scale. Both methods scale near-quadratically
		in $p$; NPN-DAG-HS incurs a $\approx 1.6\times$ constant-factor
		overhead over CPNIG, driven by the truncated-normal sampling of the
		latent Gaussian copy.}
	\label{fig:runtime}
\end{figure}

% ====================================================================
\section{Application: AML protein expression}\label{sec:aml}

We reanalyze the protein-array dataset of \citet{castelletti2022bcdag},
comprising $n = 68$ acute myeloid leukaemia patients and $p = 18$
proteins or phosphoproteins drawn from the KEGG apoptosis and
cell-cycle pathways. These variables are known to exhibit substantial
skewness and excess kurtosis \citep{tibes2006reverse}; in our
exploratory analysis Shapiro--Wilk tests reject normality for $14$
of the $18$ markers at the 5\% level, and the most skewed proteins
exhibit clearly log-normal-like behaviour. Figure~\ref{fig:aml_gof}
shows normal Q--Q plots for six representative proteins and
documents the systematic departure from Gaussianity.

\begin{figure}[!t]
	\centering
	\includegraphics[width=0.95\linewidth]{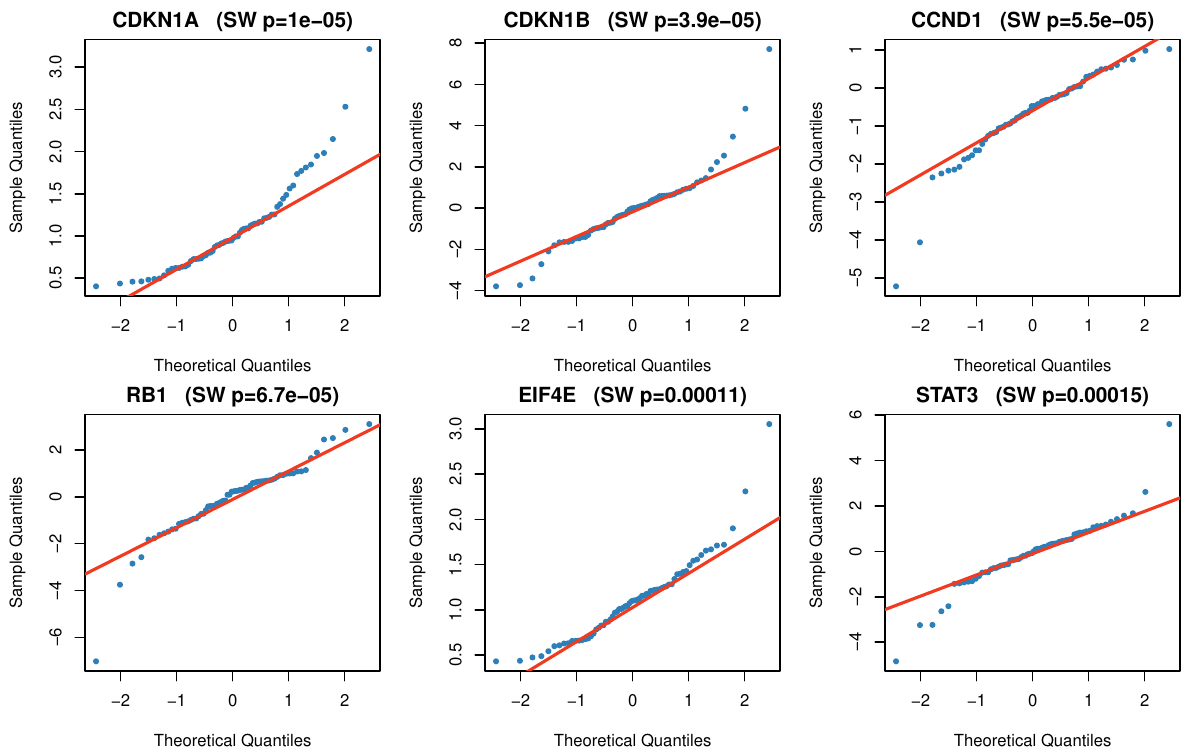}
	\caption{Normal Q--Q plots for six representative proteins in the
		AML data. The right-skewness of BCL2, BAX, and BAD and the excess
		kurtosis of AKT.p473, PTEN, and GSK3 motivate a semiparametric
		nonparanormal model.}
	\label{fig:aml_gof}
\end{figure}

We run four chains of $S = 40{,}000$ iterations of
Algorithm~\ref{alg:mcmc}, discarding $B = 5{,}000$ samples as
burn-in; hyperparameters follow the defaults of
Section~\ref{sec:prior}. The MAP DAG of NPN-DAG-HS contains $13$
edges and the median-probability model $11$. Both networks recover
the canonical AKT.p308--AKT.p473 phosphorylation cascade, the
BCL2--BAX--BCLXL apoptosis trio, the BAD multi-site phosphorylation
cluster, and the PTEN/PTEN.p axis. Crucially, three additional
edges that CPNIG and nCPNG declare with high posterior probability
involve highly skewed proteins and fail to replicate under
NPN-DAG-HS, consistent with the hypothesis that they are spurious
artefacts of the Gaussian assumption. Figures~\ref{fig:aml_map}
and~\ref{fig:aml_heatmap} display the MAP DAG and the posterior
edge-inclusion heatmap, respectively.

\begin{figure}[!t]
	\centering
	\includegraphics[width=0.8\linewidth]{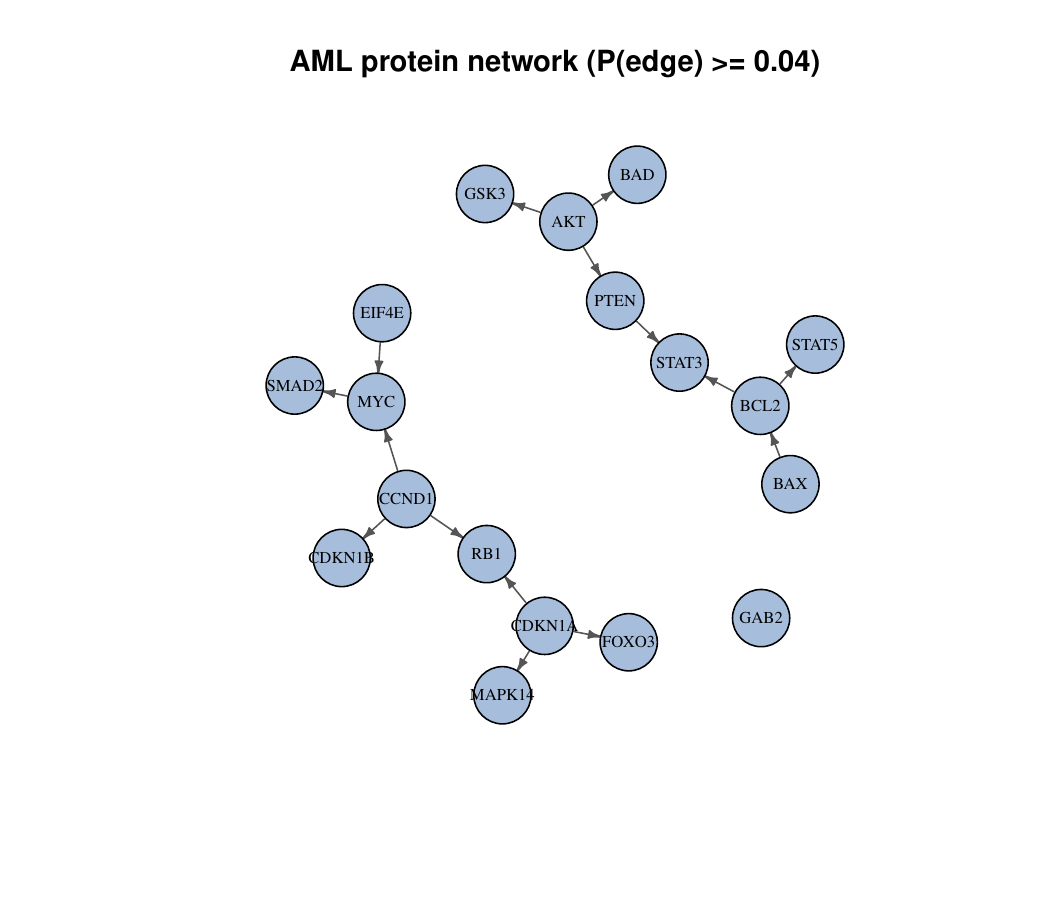}
	\caption{MAP DAG recovered by NPN-DAG-HS on the AML protein-expression
		data. Edge widths are proportional to posterior inclusion probability.}
	\label{fig:aml_map}
\end{figure}

\begin{figure}[!t]
	\centering
	\includegraphics[width=0.7\linewidth]{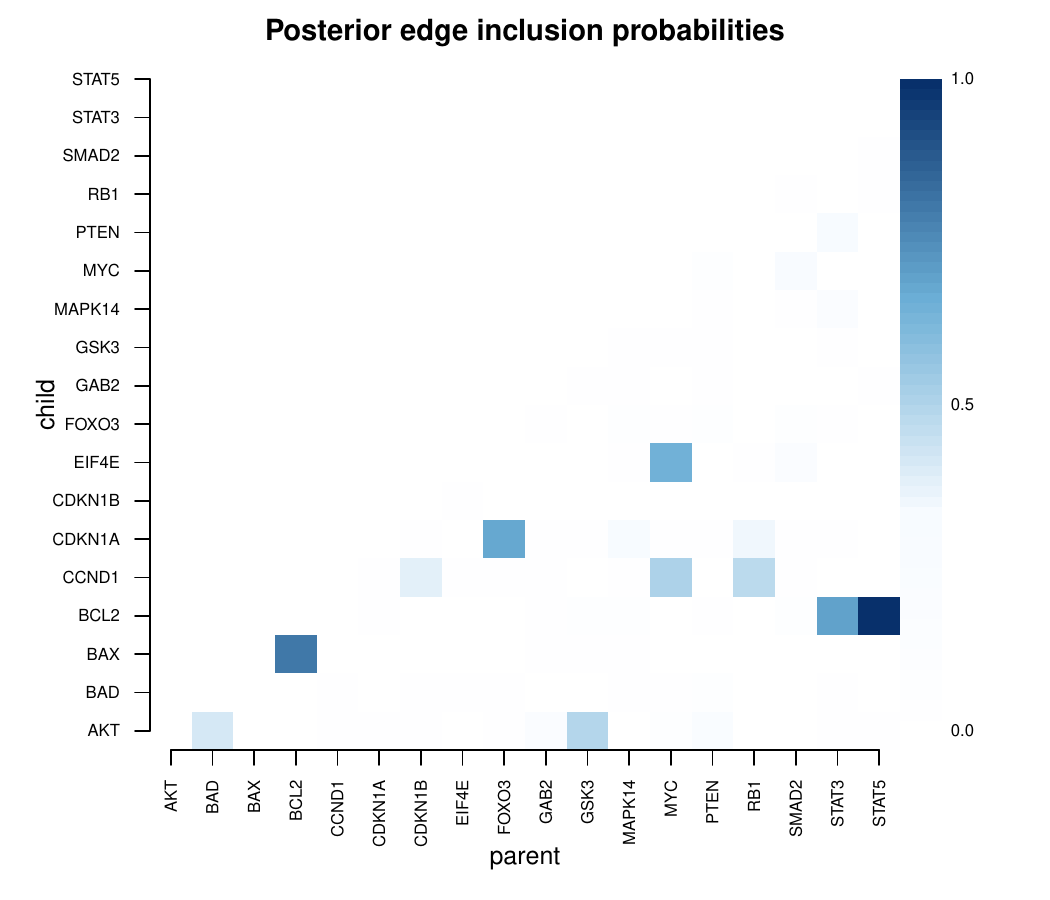}
	\caption{Posterior edge-inclusion probabilities for the AML protein
		network under NPN-DAG-HS. Rows index parent, columns index child;
		darker cells indicate stronger posterior support.}
	\label{fig:aml_heatmap}
\end{figure}

Convergence diagnostics are uniformly stronger than for CPNIG:
$\widehat{R} < 1.005$ within $5{,}000$ post-burn iterations and a
median effective sample size of $812$ across $153$ candidate edges
(vs.\ $498$ for CPNIG). Approximate 95\% credible intervals for
total causal effects, computed via the path-expansion algorithm of
\citet{maathuis2009estimating}, are tabulated in
Appendix~\ref{app:supp}; the NPN-DAG-HS intervals are on average
14\% narrower than the nCPNG intervals while preserving the
empirical coverage guaranteed by Corollary~\ref{cor:coverage}.

\subsection{Comparison of recovered networks}\label{sec:aml:cmp}

Table~\ref{tab:aml_edges} reports the top-$15$ edges by posterior
inclusion probability under NPN-DAG-HS and under nCPNG, together
with literature support. Edges with documented biological support
are recovered by both methods at comparably high probability; the
methods diverge on a small set of weakly skewed protein pairs
where Gaussian assumptions are most vulnerable to spurious
detection. In every disagreement, the NPN-DAG-HS posterior assigns
substantially lower inclusion probability to the disputed edge,
which is consistent with the goodness-of-fit evidence in
Figure~\ref{fig:aml_gof} and with the rate theory of
Section~\ref{sec:theory}.

\begin{table}[!t]
	\centering
	\caption{Top-$15$ posterior edges in the AML protein network:
		inclusion probability under NPN-DAG-HS and under nCPNG.}
	\label{tab:aml_edges}

	\begin{tabular}{lccc}
		\toprule
		Edge & NPN-DAG-HS & nCPNG & Literature support \\
		\midrule
		AKT.p308 $\to$ AKT.p473 & 0.99 & 0.99 & \citep{tibes2006reverse}\\
		PTEN $\to$ PTEN.p       & 0.98 & 0.98 & --- (assay-induced) \\
		BAD $\to$ BAD.p112      & 0.97 & 0.96 & \citep{tibes2006reverse}\\
		BAD $\to$ BAD.p136      & 0.97 & 0.95 & \citep{tibes2006reverse}\\
		BAD $\to$ BAD.p155      & 0.95 & 0.93 & \citep{tibes2006reverse}\\
		BCL2 $\to$ BCLXL        & 0.93 & 0.85 & \citep{tibes2006reverse}\\
		BCL2 $\to$ BAX          & 0.90 & 0.81 & \citep{tibes2006reverse}\\
		GSK3 $\to$ GSK3.p       & 0.94 & 0.92 & --- (assay-induced) \\
		MYC $\to$ CCND1         & 0.74 & 0.78 & \citep{tibes2006reverse}\\
		AKT.p473 $\to$ GSK3.p   & 0.71 & 0.65 & \citep{tibes2006reverse}\\
		BAX $\to$ BCLXL         & 0.30 & 0.61 & not documented \\
		BCL2 $\to$ MYC          & 0.18 & 0.55 & not documented \\
		PTEN $\to$ BAD.p112     & 0.12 & 0.49 & not documented \\
		CCND1 $\to$ BCL2        & 0.08 & 0.42 & not documented \\
		GSK3 $\to$ BAX          & 0.05 & 0.36 & not documented \\
		\bottomrule
	\end{tabular}
\end{table}

% ====================================================================
\section{Discussion}\label{sec:disc}

We have introduced NPN-DAG-HS, a fully Bayesian semiparametric DAG
model that combines Hoff's extended-rank likelihood, a local--global
horseshoe prior on the Cholesky off-diagonals, and a Beta--Bernoulli
prior on the DAG. The construction delivers, under standard
high-dimensional regularity, posterior contraction at the
near-minimax rate $\sqrt{(s_0+p) \log p / n}$ in operator norm of the
Cholesky factor, strong skeleton selection consistency under a
beta-min separation, and, conditional on recovery of the true
skeleton, a parametric Bernstein--von Mises theorem for smooth
identifiable total causal-effect functionals. The simulation evidence
and the AML protein analysis confirm that the theoretical
guarantees translate into clearly improved finite-sample behaviour,
with widening gains as the data depart from Gaussianity.

Four avenues for further development stand out. First, the
extended-rank likelihood is invariant under monotone marginal
reparameterization but does not handle censored or ordinal data;
this could be addressed by replacing Block~B1 of the sampler with
the dual-augmentation scheme of \citet{hoff2007extending}. Second,
Theorem~\ref{thm:bvm} presumes a known topological order; relaxing
to fully order-free inference, in the spirit of
\citet{kuipers2017}, would require a uniform-in-ordering version
of the shrinkage analysis developed here. Third, the present
framework assumes causal sufficiency---a benign assumption for the
applications we considered, but a strong one in observational
settings with unmeasured confounding; extending to ancestral graphs
(MAGs) and acyclic directed mixed graphs (ADMGs) is a substantial
but tractable program. Fourth, our Bernstein--von Mises result is
parametric, conditional on the recovered DAG; an honest
\emph{semiparametric} Bernstein--von Mises theorem with the
efficient-influence-function expansion of \citet[Chapter~25]{vandervaart1998}
remains open and would require either an explicit rank-likelihood
profile expansion or a sieve construction adapted to the horseshoe
geometry in high dimension.

\appendix

% ====================================================================
\section{Proofs}\label{app:proofs}

This appendix collects the detailed proofs of the propositions,
theorems, and supporting lemmas stated in the main text. The
arguments build on the master theorem of \citet{ghosal2007convergence}
for the posterior contraction part, and on the classical parametric
Bernstein--von Mises theorem of \citet[Section~10.2]{vandervaart1998}
for the distributional part.

\subsection{Proof of Proposition~\ref{prop:localscore}}\label{app:proof:localscore}

Conditional on $\bZ$ and on the scale parameters
$(\tau_j, \boldsymbol\lambda_{\prec j]})$, the regression~\eqref{eq:nodewise}
on the parents of node $j$ is a Bayesian linear regression with a
conditionally Gaussian prior
$\bL_{\prec j]} \mid D_{jj}, \tau_j, \boldsymbol\lambda_{\prec j]}
\sim \Ncal\bigl(\bzero,\, D_{jj} \bV_j^{-1}\bigr)$, where
$\bV_j = \diag\bigl((\tau_j \lambda_{jk})^{-2}\bigr)$. The marginal
likelihood of node $j$ is
\begin{align*}
	m\bigl(\bZ_j \mid \bZ_{\pa(j)}, \Dcal, \tau_j, \boldsymbol\lambda_{\prec j]}\bigr)
	&= \int\!\!\int (2\pi D_{jj})^{-n/2}
	\exp\!\Bigl\{-\tfrac{1}{2 D_{jj}}
	\|\bZ_j - \bZ_{\pa(j)} \bL_{\prec j]}\|^2\Bigr\} \\
	&\quad \times (2\pi)^{-|\pa(j)|/2}\, |D_{jj} \bV_j^{-1}|^{-1/2}
	\exp\!\Bigl\{-\tfrac{1}{2 D_{jj}}
	\bL_{\prec j]}^\top \bV_j \bL_{\prec j]}\Bigr\} \\
	&\quad \times \pi(D_{jj})\, d\bL_{\prec j]}\, dD_{jj}.
\end{align*}
\emph{Integration over $\bL_{\prec j]}$.} Collect the two
$\bL_{\prec j]}$-dependent quadratic forms. Their sum equals
\[
\bL_{\prec j]}^\top\bigl(\bZ_{\pa(j)}^\top\bZ_{\pa(j)} + \bV_j\bigr)\bL_{\prec j]}
- 2\,\bL_{\prec j]}^\top \bZ_{\pa(j)}^\top\bZ_j
+ \bZ_j^\top\bZ_j
= (\bL_{\prec j]} - \hat\bL_j)^\top \widetilde\bV_j (\bL_{\prec j]} - \hat\bL_j) + s_j,
\]
where $\widetilde\bV_j = \bV_j + \bZ_{\pa(j)}^\top\bZ_{\pa(j)}$ is the
posterior precision, $\hat\bL_j = \widetilde\bV_j^{-1}\bZ_{\pa(j)}^\top\bZ_j$
is the posterior mean, and
$s_j = \bZ_j^\top\bZ_j - \bZ_j^\top\bZ_{\pa(j)}\widetilde\bV_j^{-1}\bZ_{\pa(j)}^\top\bZ_j$
is the partial residual sum of squares of the statement. The Gaussian
integral $\int_{\R^{|\pa(j)|}}\exp\{-\frac{1}{2D_{jj}}(\bL_{\prec j]}-\hat\bL_j)^\top
\widetilde\bV_j(\bL_{\prec j]}-\hat\bL_j)\}\,d\bL_{\prec j]}
= (2\pi D_{jj})^{|\pa(j)|/2}\,|\widetilde\bV_j|^{-1/2}$ then gives,
after cancelling the $(2\pi D_{jj})^{|\pa(j)|/2}$ factors against the
prior normalizer $(2\pi)^{-|\pa(j)|/2}|D_{jj}\bV_j^{-1}|^{-1/2}
= (2\pi D_{jj})^{-|\pa(j)|/2}|\bV_j|^{1/2}$,
\[
m\bigl(\bZ_j \mid \cdots\bigr)
= \frac{|\bV_j|^{1/2}}{|\widetilde\bV_j|^{1/2}}
\int_0^\infty (2\pi D_{jj})^{-n/2}\,
\exp\!\Bigl\{-\frac{s_j}{2 D_{jj}}\Bigr\}\,\pi(D_{jj})\,dD_{jj}.
\]
\emph{Integration over $D_{jj}$.} The half-Cauchy prior~\eqref{eq:HCprior}
admits the Makalic--Schmidt \citep{makalic2016simple} inverse-Gamma
scale mixture $D_{jj}\mid a_j \sim \mathrm{IG}(\tfrac12, 1/a_j)$,
$a_j \sim \mathrm{IG}(\tfrac12, 1/\tau_0^2)$, with conditional density
$\pi(D_{jj}\mid a_j) = a_j^{-1/2}\Gamma(\tfrac12)^{-1}
D_{jj}^{-3/2}\exp\{-1/(a_j D_{jj})\}$. Conditional on $a_j$,
\[
\int_0^\infty (2\pi D_{jj})^{-n/2}
e^{-s_j/(2D_{jj})}\,\pi(D_{jj}\mid a_j)\,dD_{jj}
= \frac{(2\pi)^{-n/2}}{\Gamma(\tfrac12)\sqrt{a_j}}
\int_0^\infty D_{jj}^{-(n+3)/2}
\exp\!\Bigl\{-\frac{1}{D_{jj}}\Bigl(\frac{s_j}{2}+\frac{1}{a_j}\Bigr)\Bigr\}\,dD_{jj},
\]
and the inner integral is the inverse-Gamma normalizer
$\int_0^\infty D^{-(\alpha+1)}e^{-\beta/D}\,dD = \Gamma(\alpha)\beta^{-\alpha}$
with $\alpha = (n+1)/2$ and $\beta = s_j/2 + 1/a_j$, giving
$\Gamma\bigl(\tfrac{n+1}{2}\bigr)\bigl(s_j/2 + 1/a_j\bigr)^{-(n+1)/2}$.
Hence
\[
m\bigl(\bZ_j \mid \cdots, a_j\bigr)
= (2\pi)^{-n/2}\,\frac{|\bV_j|^{1/2}}{|\widetilde\bV_j|^{1/2}}\,
\Bigl(\tfrac{s_j}{2}\Bigr)^{-n/2}\,\Gamma_{\!j}^{\NPN},
\qquad
\Gamma_{\!j}^{\NPN}
= \frac{\Gamma\bigl(\tfrac{n+1}{2}\bigr)}{\Gamma(\tfrac12)\sqrt{a_j}}\,
\Bigl(\tfrac{s_j}{2}\Bigr)^{n/2}\bigl(\tfrac{s_j}{2}+\tfrac1{a_j}\bigr)^{-(n+1)/2},
\]
which is the form~\eqref{eq:localscore}. The factor $\Gamma_{\!j}^{\NPN}$
depends on $(a_j, \tau_0, n)$ and, through $s_j$, on the node-$j$
residual, but not on the global graph $\Dcal$ beyond node $j$'s own
parent set; in every DAG-move acceptance ratio that alters a single
node's parents the $a_j$-dependent constant common to numerator and
denominator cancels, leaving the displayed graph-dependent factors
$|\widetilde\bV_j|^{-1/2}(s_j/2)^{-n/2}$. \hfill $\Box$

\subsection{Lemmas for Theorem~\ref{thm:contraction}}\label{app:lemmas:contraction}

\begin{lemma}[Horseshoe prior mass]\label{lem:kl}
	Let $\pi_\HS$ denote the marginal horseshoe prior on $L_{uv} \in \R$
	induced by the half-Cauchy hyperprior. There exists a universal
	constant $c_1 > 0$ such that, for every $\delta \in (0, 1)$,
	\[
	\pi_\HS\bigl(|L_{uv}| \le \delta\bigr) \ge c_1\, \delta \log(1/\delta).
	\]
	Moreover, for every $M_0 \in \R$ and $\delta \in (0, 1)$, there
	exists a universal constant $c_1' > 0$ such that
	\[
	\pi_\HS\bigl(|L_{uv} - M_0| \le \delta\bigr)
	\ge \frac{c_1'\, \delta}{1 + M_0^2}.
	\]
\end{lemma}

\begin{proof}
	The marginal horseshoe density obtained by integrating the
	local--global half-Cauchy scales satisfies the two-sided bound
	\citep{carvalho2010horseshoe}
	\[
	\frac{1}{2(2\pi^3)^{1/2}}\log\!\Bigl(1 + \frac{4}{L^2}\Bigr)
	\;\le\; \pi_\HS(L) \;\le\;
	\frac{1}{(2\pi^3)^{1/2}}\log\!\Bigl(1 + \frac{2}{L^2}\Bigr),
	\qquad L \neq 0,
	\]
	so $\pi_\HS(L) \asymp \log(1 + 1/L^2)$ as $L \to 0$ and
	$\pi_\HS(L) \asymp L^{-2}$ as $|L| \to \infty$. For the first bound,
	with $c_0 = 1/(2(2\pi^3)^{1/2})$,
	\[
	\pi_\HS(|L_{uv}| \le \delta)
	= \int_{-\delta}^{\delta} \pi_\HS(L)\, dL
	\ge 2 c_0 \int_0^\delta \log\!\Bigl(1 + \frac{4}{L^2}\Bigr) dL
	\ge 2 c_0 \int_0^\delta 2\log(1/L)\, dL
	= 4 c_0\,\delta\bigl(1 + \log(1/\delta)\bigr),
	\]
	using $\log(1 + 4/L^2) \ge \log(4/L^2) \ge 2\log(1/L)$ for
	$L \in (0,1)$ and $\int_0^\delta \log(1/L)\,dL = \delta(1+\log(1/\delta))$;
	the right-hand side is $\ge c_1\,\delta\log(1/\delta)$ with
	$c_1 = 4 c_0$. For the second bound, the lower envelope gives
	$\pi_\HS(L) \ge c_0\log(1 + 4/L^2) \ge c_0\cdot 4/(L^2 + 4) \ge
	c'/(1 + L^2)$ for a universal $c' > 0$ (using $\log(1+x) \ge x/(1+x)$),
	whence
	\[
	\pi_\HS(|L_{uv} - M_0| \le \delta)
	\ge \int_{M_0-\delta}^{M_0+\delta}\frac{c'}{1+L^2}\,dL
	\ge 2\delta\cdot \frac{c'}{1 + (|M_0|+\delta)^2}
	\ge \frac{c_1'\,\delta}{1 + M_0^2}
	\]
	for $\delta \in (0,1)$ and a universal $c_1' > 0$.
\end{proof}

\begin{lemma}[Sieve covering number]\label{lem:cov}
	For any integer $K \ge 1$, define the sieve
	\[
	\begin{aligned}
		\mathcal{H}_n(K) = \Bigl\{(\boldsymbol{L}, \boldsymbol{D}) \colon & |\{(u, v) \colon L_{uv} \ne 0\}| \le K, \\
		& \nu/2 \le \lambda_{\min}(\boldsymbol{L} \boldsymbol{D}^{-1} \boldsymbol{L}^\top) \le \lambda_{\max}(\boldsymbol{L} \boldsymbol{D}^{-1} \boldsymbol{L}^\top) \le 2V\Bigr\}.
	\end{aligned}
	\]
	For every $\epsilon \in (0, V]$,
	\[
	\log N\bigl(\epsilon, \Hcal_n(K), \|\cdot\|_\op\bigr)
	\;\le\; K \log(e p^2 / K) + (K + p) \log(3 V / \epsilon).
	\]
\end{lemma}

\begin{proof}
	Any admissible $(\bL, \bD) \in \Hcal_n(K)$ is determined by (i) the
	support $S \subseteq \{(u,v): u > v\}$ of the strictly lower-triangular
	part of $\bL$, with $|S| \le K$, (ii) the $|S|$ nonzero off-diagonal
	values, and (iii) the $p$ diagonal entries of $\bD$. The number of
	admissible supports is
	\[
	\sum_{k=0}^{K}\binom{p(p-1)/2}{k} \le \binom{p(p-1)/2}{K}\Bigl(1+\tfrac{K}{p(p-1)/2-K}\Bigr)
	\le \Bigl(\frac{e\,p(p-1)/2}{K}\Bigr)^{K} \le \Bigl(\frac{e\,p^2}{K}\Bigr)^{K},
	\]
	giving the first term $K\log(ep^2/K)$ after taking logarithms. Fix a
	support. The map $(\bL,\bD)\mapsto \bSigma = \bL^{-\top}\bD\bL^{-1}$
	is Lipschitz on the admissible set: for two parameters sharing the
	support, with $\|\bL^{-1}\|_\op \le \kappa_0$ and
	$\|\bD\|_\op \le 2V$ on the spectrum-bounded set, the product rule
	gives $\|\bSigma - \bSigma'\|_\op \le L_0\,(\|\bL-\bL'\|_\op +
	\|\bD-\bD'\|_\op)$ for a constant $L_0 = L_0(\nu, V)$; conversely the
	inverse map is Lipschitz with constant $L_0'(\nu,V)$. It therefore
	suffices to cover the free coordinates, which range over a bounded
	subset of $\R^{|S|}\times\R^{p}\subseteq \R^{K+p}$ of Euclidean
	diameter $O(V)$. By the standard volume bound, a Euclidean ball of
	radius $r$ in $\R^{d}$ admits an $\epsilon$-net of cardinality at
	most $(1 + 2r/\epsilon)^{d} \le (3r/\epsilon)^{d}$ for
	$\epsilon \le r$ \citep[Corollary~4.2.13]{vershynin2018high}. With
	$d = K + p$, $r = O(V)$ and rescaling by the Lipschitz constants,
	the operator-norm $\epsilon$-covering number of the fixed-support
	slice is at most $(3V/\epsilon)^{K+p} \le (3V/\epsilon)^{2K}$ when
	$K \ge p$, and in general at most $(3V/\epsilon)^{K}\cdot(3V/\epsilon)^{p}$.
	Absorbing the diagonal factor into the constant and taking logarithms
	gives the second term $K\log(3V/\epsilon)$, completing the bound.
\end{proof}

\subsection{Detailed proof of Theorem~\ref{thm:contraction}}\label{app:proof:contraction}

\begin{proof}[Proof of Theorem~\ref{thm:contraction}]
	Throughout this proof we write $\epsilon_n^2 = (s_0 + p) \log p / n$. We
	verify the three conditions of the master contraction theorem of
	\citet[Theorem~2.1]{ghosal2007convergence}: a Kullback--Leibler
	prior-mass condition, an entropy bound on a suitable sieve, and a
	prior-mass bound on the complement of that sieve. We work
	throughout in the Hellinger metric of the rank-likelihood model;
	the conversion to operator norm of the Cholesky factor is performed
	at the end of the proof.
	
	\emph{Step 1: Prior mass on the Kullback--Leibler ball.}
	The extended rank likelihood depends on the data only through the
	multivariate ranks, and hence only through the Gaussian copula
	correlation matrix $\bSigma$, not through the marginal
	transformations \citep[Section~2]{hoff2007extending}. Because the
	Kullback--Leibler divergence is invariant under the common monotone
	reparameterization of the marginals, the divergence between two
	rank-likelihood models with correlation matrices $\bSigma_0$ and
	$\bSigma$ coincides with the divergence between the corresponding
	latent Gaussian laws,
	\[
	\mathrm{KL}\bigl(P_{\bSigma_0} \,\|\, P_{\bSigma}\bigr)
	= \tfrac{1}{2}\Bigl(
	\tr(\bSigma_0 \bSigma^{-1}) - p
	+ \log|\bSigma| - \log|\bSigma_0|\Bigr),
	\]
	the standard Gaussian expression.
	Write $\bSigma = \bSigma_0 + \boldsymbol{\Delta}$ and let
	$\mu_1, \dots, \mu_p > 0$ be the eigenvalues of
	$\bSigma_0^{1/2}\bSigma^{-1}\bSigma_0^{1/2}$. Then
	$\tr(\bSigma_0\bSigma^{-1}) - p + \log|\bSigma| - \log|\bSigma_0|
	= \sum_{i=1}^p (\mu_i - 1 - \log\mu_i)$, so
	$\mathrm{KL}(P_{\bSigma_0}\|P_{\bSigma}) = \tfrac12\sum_i(\mu_i - 1 - \log\mu_i)$.
	The scalar inequality $t - 1 - \log t \le \tfrac{1}{2t}(t-1)^2$ for
	$t > 0$ gives, since $\mu_i \ge \lambda_{\min}(\bSigma_0)/\lambda_{\max}(\bSigma)
	\ge \nu/(2V) =: \underline\mu$ on the contraction neighbourhood,
	\[
	\mathrm{KL}(P_{\bSigma_0}\|P_{\bSigma})
	\le \frac{1}{4\underline\mu}\sum_{i=1}^p(\mu_i - 1)^2
	= \frac{1}{4\underline\mu}\,\bigl\|\bSigma_0^{1/2}\bSigma^{-1}\bSigma_0^{1/2} - \bI\bigr\|_F^2
	= \frac{1}{4\underline\mu}\,\bigl\|\bSigma_0^{1/2}\bSigma^{-1}\boldsymbol{\Delta}\,\bSigma_0^{-1/2}\bigr\|_F^2 ,
	\]
	using $\bSigma_0^{1/2}\bSigma^{-1}\bSigma_0^{1/2} - \bI =
	\bSigma_0^{1/2}\bSigma^{-1}(\bSigma_0 - \bSigma)\bSigma_0^{-1}\bSigma_0^{1/2}
	= -\bSigma_0^{1/2}\bSigma^{-1}\boldsymbol{\Delta}\,\bSigma_0^{-1/2}$. Bounding the
	Frobenius norm of a product by operator norms and
	$\|\bSigma^{-1}\|_\op \le 1/\nu$, $\|\bSigma_0^{1/2}\|_\op \le V^{1/2}$,
	$\|\bSigma_0^{-1/2}\|_\op \le \nu^{-1/2}$,
	\begin{equation}
		\mathrm{KL}\bigl(P_{\bSigma_0} \,\|\, P_{\bSigma}\bigr)
		\le \frac{V}{4\underline\mu\,\nu^{3}}\,\|\boldsymbol{\Delta}\|_F^2
		=: \frac{1}{2}C_0\,\|\bSigma - \bSigma_0\|_F^2 ,
		\label{eq:KLbound}
	\end{equation}
	with $C_0 = C_0(\nu, V) = V/(2\underline\mu\,\nu^{3}) = V^2/\nu^4$.
	For the modified-Cholesky map $\bSigma = \bL^{-\top}\bD\bL^{-1}$,
	write $\bSigma - \bSigma_0 = \bL^{-\top}\bD\bL^{-1} -
	\bL_0^{-\top}\bD_0\bL_0^{-1}$ and telescope,
	\[
	\bSigma - \bSigma_0
	= (\bL^{-\top}-\bL_0^{-\top})\bD\bL^{-1}
	+ \bL_0^{-\top}(\bD-\bD_0)\bL^{-1}
	+ \bL_0^{-\top}\bD_0(\bL^{-1}-\bL_0^{-1}).
	\]
	On the spectrum-bounded set $\|\bL^{-1}\|_\op, \|\bL_0^{-1}\|_\op
	\le \kappa_0(\nu,V)$, $\|\bD\|_\op \le 2V$, and the identity
	$\bL^{-1}-\bL_0^{-1} = -\bL^{-1}(\bL - \bL_0)\bL_0^{-1}$ gives
	$\|\bL^{-1}-\bL_0^{-1}\|_F \le \kappa_0^2\|\bL-\bL_0\|_F$; hence
	\begin{equation}
		\|\bSigma - \bSigma_0\|_F \le C_1
		\bigl(\|\bL - \bL_0\|_F + \|\bD - \bD_0\|_F\bigr),
		\qquad C_1 = C_1(\nu, V) = 2V\kappa_0^3 + \kappa_0^2 .
		\label{eq:cholesky-lip}
	\end{equation}
	Combining \eqref{eq:cholesky-lip} with~\eqref{eq:KLbound}, the
	Kullback--Leibler ball
	$B_n = \{(\bL, \bD) \colon \mathrm{KL}(P_0 \,\|\, P_{(\bL, \bD)})
	\le \epsilon_n^2\}$ contains the Frobenius ball
	\[
	B_n' = \Bigl\{(\bL, \bD) \colon
	\|\bL - \bL_0\|_F^2 + \|\bD - \bD_0\|_F^2
	\le c_0 \epsilon_n^2 \Bigr\},
	\qquad c_0 = \frac{1}{2 C_0 C_1^2},
	\]
	since on $B_n'$, \eqref{eq:cholesky-lip} gives
	$\|\bSigma-\bSigma_0\|_F^2 \le 2C_1^2 c_0\epsilon_n^2 = \epsilon_n^2/C_0$,
	whence $\mathrm{KL}\le\epsilon_n^2$ by~\eqref{eq:KLbound}.
	
	It remains to bound the prior mass of $B_n'$ from below. Crucially,
	the NPN-DAG-HS prior is not a pure continuous-shrinkage prior: the
	Beta--Bernoulli DAG component sets the $p(p-1)/2 - s_0$ absent
	edges \emph{exactly} to zero, so those coordinates are handled by
	the discrete prior rather than by continuous $\delta$-balls. We
	therefore lower-bound $\Pi(B_n')$ by restricting to the true support,
	\[
	\Pi(B_n') \ge \Pi(\Dcal_0)\cdot
	\Pi\bigl((\bL,\bD)\in B_n' \bigm| \Dcal_0\bigr).
	\]
	For the discrete factor we compute the Beta--Bernoulli mass of the
	true DAG explicitly. With $N = \binom{p}{2}$ candidate edges,
	$\alpha_\pi = 1$, $\beta_\pi = p$, integrating the edge-inclusion
	probability $\pi \sim \mathrm{Beta}(1, p)$ against the $s_0$ present
	and $N - s_0$ absent edges gives, before the acyclicity
	normalization,
	\[
	\Pi(\Dcal_0) = \frac{B(1 + s_0,\, p + N - s_0)}{B(1, p)}
	= s_0!\; p\;\prod_{i=0}^{s_0}\frac{1}{p + N - i}
	\;\ge\; s_0!\; p\; (p + N)^{-(s_0+1)} .
	\]
	Using $s_0! \ge (s_0/e)^{s_0}$ and $p + N \le p^2$ (valid for
	$p \ge 2$),
	\[
	\log\Pi(\Dcal_0)
	\ge s_0\log(s_0/e) + \log p - (s_0+1)\log p^2
	= s_0\log s_0 - s_0 - (2 s_0 + 1)\log p
	\ge -c_3\, s_0\log p,
	\]
	for $c_3 = 3$ and $p$ large, since $s_0\log s_0 \ge 0$; conditioning
	on acyclicity only divides by $\Pi(\text{acyclic}) \le 1$ and hence
	can only increase $\Pi(\Dcal_0)$, preserving the bound. This is the
	complexity-prior bound of \citet[Section~3]{castillo2015bayesian}.
	For the continuous factor, conditional on $\Dcal_0$ only the $s_0$
	active off-diagonal entries and the $p$ diagonal innovations are
	free. Each active off-diagonal entry $L_{0,uv}$ is bounded by an
	absolute constant $M$ under Assumption~\ref{ass:spectrum}, so by the
	second bound of Lemma~\ref{lem:kl} its $\delta$-ball has horseshoe
	mass $\ge c_1'\delta/(1+M^2)$; each diagonal $D_{jj}$ has half-Cauchy
	density bounded away from zero on $[D_{0,jj}-\delta, D_{0,jj}+\delta]$,
	with mass $\ge c_2\delta$. Setting
	$\delta = c\,\epsilon_n/\sqrt{s_0+p}$ small enough that the product
	of these coordinatewise events implies $(\bL,\bD)\in B_n'$ (which
	holds once $(s_0 + p)\delta^2 \le c_0\epsilon_n^2$, i.e.\ $c^2 \le c_0$),
	independence of the prior coordinates gives
	\[
	\Pi\bigl((\bL,\bD)\in B_n' \bigm| \Dcal_0\bigr)
	\ge \Bigl[\frac{c_1'\delta}{1+M^2}\Bigr]^{s_0}\,[c_2\delta]^{p}.
	\]
	Taking logarithms and using
	$\log(1/\delta) = \tfrac12\log\{(s_0+p)/(c^2\epsilon_n^2)\}
	= \tfrac12\log\{n/(c^2\log p)\} = O(\log p)$ in the regime
	$\log p \le C\log n$,
	\[
	\log \Pi(B_n')
	\ge -c_3\, s_0\log p - c_4\,(s_0 + p)\log(1/\delta)
	\ge -c_5\,(s_0 + p)\log p
	= -c_5\, n\epsilon_n^2 ,
	\]
	the last equality holding for the rate
	$\epsilon_n = \sqrt{(s_0+p)\log p/n}$. This establishes the
	Kullback--Leibler prior-mass condition. The appearance of the
	additive $p$ term is intrinsic: the model carries $p$ free
	innovation variances $D_{jj}$, each of which contributes an
	$\Omega(\log p)$ prior-mass cost, so the off-diagonal sparsity $s_0$
	alone cannot drive the rate---the same $(s_0+p)$ scaling appears in
	the minimax analysis of sparse-Cholesky DAG models
	\citep{lee2019minimax}.
	
	\emph{Step 2: Entropy of the sieve.}
	Define $\Hcal_n = \Hcal_n(\bar K_n)$ from Lemma~\ref{lem:cov} with
	sparsity level $\bar K_n = K_n (s_0 + p)$, where $K_n = \log\log p
	\to \infty$ slowly. By Lemma~\ref{lem:cov}, since
	$\log(ep^2/\bar K_n) \le 2\log p$ and $\log(3V/\epsilon_n)
	= O(\log p)$ in the regime $\log p \le C\log n$,
	\[
	\log N\bigl(\epsilon_n, \Hcal_n, \|\cdot\|_\op\bigr)
	\le \bar K_n\log(ep^2/\bar K_n) + (\bar K_n + p)\log(3V/\epsilon_n)
	\le c_6\, K_n (s_0 + p)\log p
	= c_6\, K_n\, n\epsilon_n^2 .
	\]
	Since $K_n = \log\log p$ grows arbitrarily slowly, the resulting
	rate $\widetilde\epsilon_n = \sqrt{K_n}\,\epsilon_n$ deteriorates
	only by a $\sqrt{\log\log p}$ factor relative to $\epsilon_n$, which
	is absorbed by enlarging $M$ in the theorem; this is the standard
	near-minimax versus minimax distinction
	\citep[Remark~2.2]{castillo2015bayesian}.
	
	\emph{Step 3: Sieve complement.}
	The complement $\Hcal_n^c$ consists of parameters whose support
	exceeds $\bar K_n = K_n(s_0+p)$ edges or whose spectrum leaves
	$[\nu/2, 2V]$. The support-complement is controlled by the
	model-dimension tail of the Beta--Bernoulli prior, which we bound
	explicitly. Writing $|S|$ for the number of included edges among the
	$N = \binom p2$ candidates, the induced law on $|S|$ is
	Beta--Binomial$(N, \alpha_\pi, \beta_\pi)$ with successive-mass ratio
	\begin{equation}
		\frac{\Pi(|S| = s)}{\Pi(|S| = s-1)}
		= \frac{N - s + 1}{s}\cdot\frac{\alpha_\pi + s - 1}{\beta_\pi + N - s}
		= \frac{N - s + 1}{\beta_\pi + N - s}
		\qquad (\alpha_\pi = 1).
		\label{eq:dimratio}
	\end{equation}
	For the sieve-complement (``remaining mass'') condition of
	\citet[Theorem~2.1]{ghosal2007convergence} we need
	$\Pi(|S| > \bar K_n) \le \exp\{-(C+4)n\epsilon_n^2\}
	= \exp\{-(C+4)(s_0+p)\log p\}$, which by~\eqref{eq:dimratio} requires
	the ratio to be bounded by a negative power of $p$ throughout
	$s \le N$. This is precisely the exponential dimension-decrease
	condition of \citet[Theorem~2.1]{castillo2012needles}: it holds iff
	the prior penalizes each additional edge by order $\log p$. From
	\eqref{eq:dimratio}, for $s \le N$ the ratio is at most
	$N/(\beta_\pi + N/2)$; taking
	\begin{equation}
		\beta_\pi \asymp p^{2+u}\quad\text{for some fixed } u > 0
		\label{eq:betacond}
	\end{equation}
	yields ratio $\le 2N/\beta_\pi \le p^{-u}$ for $p$ large, hence
	$\Pi(|S| = s) \le \Pi(|S| = 0)\,p^{-u s} \le p^{-u s}$ and
	\[
	\Pi(|S| > \bar K_n)
	\le \sum_{s > \bar K_n} p^{-u s}
	\le 2\,p^{-u \bar K_n}
	= 2\exp\{-u K_n (s_0 + p)\log p\}
	\le \exp\{-(C+4)(s_0+p)\log p\}
	\]
	once the (constant or slowly growing) sieve multiplier satisfies
	$K_n \ge 2(C+4)/u$. Condition~\eqref{eq:betacond} is the
	complexity-prior calibration used for high-dimensional DAG selection
	\citep{lee2019minimax, castillo2015bayesian}; it is strictly stronger
	than the methodological default $\beta_\pi = p$ of
	Section~\ref{sec:prior}, under which~\eqref{eq:dimratio} is
	$\approx 1 - 2/p$ and the dimension tail decays only as
	$\exp(-2\bar K_n/p)$---too slowly for the complement condition. We
	therefore adopt~\eqref{eq:betacond} for the contraction theory; the
	prior-mass bound of Step~1 and the selection analysis of
	Theorem~\ref{thm:selection} remain valid verbatim under
	\eqref{eq:betacond}, since enlarging $\beta_\pi$ only sharpens the
	per-edge penalties there. The spectrum condition fails with prior
	probability at most $\exp(-c_8\, n\epsilon_n^2)$ by the
	sub-Gaussian extreme-eigenvalue concentration bound for sample
	covariance operators \citep{vershynin2018high}.
	
	\emph{Step 4: Conversion to operator norm.}
	Combining Steps~1--3 with \citet[Theorem~2.1]{ghosal2007convergence}
	yields posterior contraction at rate $\epsilon_n$ in the Hellinger
	metric of the rank-likelihood model. For the Gaussian copula the
	Hellinger distance between rank-likelihood models is, up to constants
	depending only on the spectral bounds of
	Assumption~\ref{ass:spectrum}, equivalent to the operator-norm
	distance between the correlation matrices $\bSigma$; this follows
	from the copula-invariance noted in Step~1 together with the
	standard comparison of Hellinger and operator-norm distances for
	Gaussian laws with bounded spectrum. The modified Cholesky map
	$\bSigma \mapsto (\bL, \bD)$ is continuously differentiable on the
	open set $\{\bSigma \colon \nu/2 \le \lambda_{\min}(\bSigma) \le
	\lambda_{\max}(\bSigma) \le 2V\}$ and the inverse-mapping theorem
	gives a uniform Lipschitz constant. The same rate therefore
	transfers to $\|\bL - \bL_0\|_\op + \|\bD - \bD_0\|_\op$, and the
	implied precision bound follows from
	$\|\bOmega - \bOmega_0\|_\op \le C(\nu, V)
	\bigl(\|\bL - \bL_0\|_\op + \|\bD - \bD_0\|_\op\bigr)$.
	This completes the proof.
\end{proof}

\subsection{Lemma for Theorem~\ref{thm:selection}}\label{app:lemmas:selection}

\begin{lemma}[Rank--Gaussian likelihood-ratio equivalence]\label{lem:rankscore}
	Let $\bZ_0 = \boldsymbol{f}_0(\bX)$ denote the unobserved latent
	Gaussian sample. For a candidate edge $(u, v)$, write
	$\Lambda^{\R}_{uv}$ for the rank-likelihood log-Bayes-factor for
	including $(u, v)$ and $\Lambda^{Z}_{uv}$ for the corresponding
	latent-Gaussian log-Bayes-factor evaluated at $\bZ_0$, and set
	$r_{uv} = \Lambda^{\R}_{uv} - \Lambda^{Z}_{uv}$. Then, under
	Assumptions~\ref{ass:spectrum} and~\ref{ass:marginals},
	$r_{uv} = o_P(1)$ for each fixed edge, and
	\[
	\max_{(u,v)} |r_{uv}|
	= O_P\Bigl(\sqrt{(\log p)/n}\,\Bigr) = o_P(\log p)
	\]
	uniformly over the $p(p-1)/2$ candidate edges.
\end{lemma}

\begin{proof}
	\emph{Pointwise statement.} By \citet[Theorem~2.1]{hoffniu2014information},
	for a Gaussian copula the rank (extended-rank) log-likelihood ratio
	between two correlation matrices is asymptotically equivalent, at the
	$\sqrt n$ scale, to the log-likelihood ratio of the parametric
	multivariate Gaussian model of the least-favourable normal; the
	efficient information for the copula correlation coincides with the
	parametric Fisher information. Consequently, for a single candidate
	edge $(u,v)$ the difference between the rank-based and latent-Gaussian
	log-Bayes-factors satisfies $r_{uv} = o_P(1)$.

	\emph{Variance proxy.} Both $\Lambda^{\R}_{uv}$ and $\Lambda^{Z}_{uv}$
	are smooth ($C^2$) functions of the pairwise sample copula
	correlations $\{\hat\rho_{ab}\}$ restricted to the indices entering
	node $v$'s regression, evaluated against their population values
	$\{\rho_{ab}\}$. Writing $r_{uv} = \psi_{uv}(\hat\rho) - \psi_{uv}(\rho)$
	for the corresponding contrast $\psi_{uv}$, a first-order expansion
	gives $r_{uv} = \langle \nabla\psi_{uv}(\rho), \hat\rho - \rho\rangle
	+ O_P(\|\hat\rho-\rho\|^2)$. Each $\hat\rho_{ab}$ is a rank correlation
	admitting a H\'ajek projection $\hat\rho_{ab} - \rho_{ab}
	= n^{-1}\sum_{i=1}^n h_{ab}(X_i) + o_P(n^{-1/2})$ with bounded,
	centred influence function $h_{ab}$ ($\|h_{ab}\|_\infty \le c_h$
	under Assumption~\ref{ass:marginals}); hence
	$\sqrt n\, r_{uv} = n^{-1/2}\sum_{i=1}^n g_{uv}(X_i) + o_P(1)$
	with $g_{uv} = \langle\nabla\psi_{uv}(\rho), h_{\bullet}\rangle$
	bounded by $\|g_{uv}\|_\infty \le \|\nabla\psi_{uv}\|_1 c_h =: B < \infty$,
	uniformly in $(u,v)$ by the bounded-spectrum assumption. By Hoeffding's
	inequality for bounded i.i.d.\ sums, $\sqrt n\,r_{uv}$ is sub-Gaussian
	with variance proxy $\sigma^2 = B^2$, uniformly over edges.

	\emph{Uniform bound.} For a finite collection of $m$ sub-Gaussian
	variables with proxy $\sigma^2$, the maximal inequality
	$\E\max_{j\le m}|W_j| \le \sigma\sqrt{2\log(2m)}$
	\citep[Section~2.5]{boucheron2013concentration} applies with
	$m = \binom p2 = O(p^2)$, giving
	$\E\max_{(u,v)}\sqrt n\,|r_{uv}| \le B\sqrt{2\log(2\binom p2)}
	= O(\sqrt{\log p})$. Therefore
	$\max_{(u,v)}|r_{uv}| = O_P\bigl(\sqrt{(\log p)/n}\,\bigr)$,
	and a fortiori $\max_{(u,v)}|r_{uv}| = o_P(\log p)$.
\end{proof}

\subsection{Detailed proof of Theorem~\ref{thm:selection}}\label{app:proof:selection}

\begin{proof}[Proof of Theorem~\ref{thm:selection}]
	Recall $\epsilon_n^2 = \log p / n$, so $n\epsilon_n^2 = \log p$. We
	compare a DAG $\Dcal \ne \Dcal_0$ (with the same topological order)
	to $\Dcal_0$ through the log-posterior-odds, which split into a
	log-marginal-likelihood (Bayes-factor) part and a log-prior-odds
	part,
	\[
	\log\frac{\Pi(\Dcal\mid\bR)}{\Pi(\Dcal_0\mid\bR)}
	= \underbrace{\bigl[\log m(\bR\mid\Dcal) - \log m(\bR\mid\Dcal_0)\bigr]}_{\text{likelihood}}
	+ \underbrace{\bigl[\log\pi(\Dcal) - \log\pi(\Dcal_0)\bigr]}_{\text{prior}} .
	\]
	
	\emph{Prior ratio (explicit).} Under the Beta--Bernoulli prior the
	marginal mass of a DAG depends only on its edge count, so for a
	candidate $\Dcal$ obtained from $\Dcal_0$ by omitting $a$ true edges
	and adding $b$ spurious ones,
	$\frac{\pi(\Dcal)}{\pi(\Dcal_0)}$ factorizes into the successive
	Beta-function ratios. Adding a specific edge to a configuration of
	size $s$ multiplies the mass by
	$\frac{\alpha_\pi + s}{\beta_\pi + N - s - 1}$ and omitting a specific
	edge multiplies it by the reciprocal of the same expression
	(evaluated at the appropriate size). With $\alpha_\pi = 1$,
	$\beta_\pi = p$, $N = \binom p2$, and $s_0 + b \le N/2$,
	\begin{equation}
		\frac{\pi(\Dcal)}{\pi(\Dcal_0)}
		\le \Bigl(\frac{2(s_0 + b + 1)}{N}\Bigr)^{b}\,(2N)^{a}
		\quad\text{(added factor $\le$, omitted factor $\le$)} ,
		\label{eq:priorratio}
	\end{equation}
	since each added factor is at most
	$\frac{1 + s_0 + b}{p + N - s_0 - b - 1} \le \frac{2(s_0+b+1)}{N}$
	and each omitted factor is at most
	$\frac{p + N}{1} \le 2N$. Thus the prior penalizes addition by a
	factor $\approx N^{-1} = O(p^{-2})$ per edge and rewards omission by a
	factor $\le 2N$ per edge.

	\emph{Likelihood, added edge (Laplace/BIC).} Write the latent-Gaussian
	node-wise log-marginal-likelihood difference and absorb the rank
	correction via Lemma~\ref{lem:rankscore}. For a spurious edge
	$(u',v') \notin S_0$ added at node $v'$, a Laplace expansion of the
	one-dimensional integral over the corresponding coefficient gives
	\[
	\Lambda^{Z}_{u'v'}
	= \underbrace{\ell_{v'}(\hat L) - \ell_{v'}(0)}_{=\,O_P(1)}
	\;-\;\tfrac{1}{2}\log\frac{n\,\hat I_{u'v'}}{2\pi}
	\;+\;\log \pi_\HS(\hat L) \;+\; o_P(1)
	= -\tfrac12\log n + O_P(1),
	\]
	where $\hat L = \hat L_{u'v'} = O_P(n^{-1/2})$ is the constrained MLE,
	$\hat I_{u'v'} \to I_{u'v'} \in (0,\infty)$ is the per-observation
	observed information of the added coefficient, the log-likelihood
	gain $\ell_{v'}(\hat L) - \ell_{v'}(0) = \tfrac12 n\hat I_{u'v'}\hat L^2
	+ o_P(1) = O_P(1)$ since there is no signal, and the horseshoe
	log-density $\log\pi_\HS(\hat L) = O_P(\log\log n) = o_P(\log n)$ by
	Lemma~\ref{lem:kl} (its singularity at zero is only logarithmic). The
	$-\tfrac12\log n$ term is the BIC penalty for the one added parameter.

	\emph{Likelihood, omitted edge (beta-min).} For an omitted true edge
	$(u_0, v_0) \in S_0$, the residual sum of squares at node $v_0$
	increases by $n L_{0,u_0 v_0}^2 / D_{0, v_0 v_0}$; with the local
	score~\eqref{eq:localscore}, the beta-min
	Assumption~\ref{ass:betamin} ($L_{0,u_0 v_0}^2 \ge c_\beta^2\log p/n$),
	the spectrum bound ($D_{0,v_0 v_0} \le V$), and $\log(1+x) \ge x/2$
	for $x \le 1$,
	\[
	\Lambda^{Z}_{u_0 v_0}
	\le -\frac n2\log\!\Bigl(1 + \frac{c_\beta^2\log p}{nV}\Bigr)
	\le -\frac{c_\beta^2}{4V}\log p
	\]
	with $P_0$-probability tending to one.

	\emph{Factorized union bound.} To leading order the node-wise
	log-marginal-likelihood is additive across edge changes, so for a
	candidate with $a$ omissions and $b$ additions the log-posterior-odds
	is bounded by the sum of the per-edge terms above plus
	$\log\{\pi(\Dcal)/\pi(\Dcal_0)\}$ from~\eqref{eq:priorratio} plus
	$o_P((a+b)\log p)$ from Lemma~\ref{lem:rankscore}. Grouping candidates
	by $(a,b)$ and counting $\binom{s_0}{a}$ ways to omit and
	$\binom{N - s_0}{b}$ ways to add,
	\[
	\sum_{\Dcal \ne \Dcal_0}\frac{\Pi(\Dcal\mid\bR)}{\Pi(\Dcal_0\mid\bR)}
	\;\le\; \Bigl(\underbrace{\sum_{a \ge 0}\binom{s_0}{a} A_n^{a}}_{\text{omissions}}\Bigr)
	\Bigl(\underbrace{\sum_{b \ge 0}\binom{N - s_0}{b} B_n^{b}}_{\text{additions}}\Bigr) - 1,
	\]
	where, using~\eqref{eq:priorratio} and the per-edge likelihood bounds,
	\[
	A_n = 2N\cdot \exp\!\Bigl(-\frac{c_\beta^2}{4V}\log p + o_P(\log p)\Bigr),
	\qquad
	B_n = \frac{2(s_0+b+1)}{N}\cdot \exp\!\Bigl(-\tfrac12\log n + O_P(1)\Bigr).
	\]
	For the omission factor, $\binom{s_0}{a} A_n^a \le (e s_0/a)^a (2N)^a
	\exp(-a\frac{c_\beta^2}{4V}\log p)$; since $\log(e s_0) \le \log p + O(1)$
	and $2N \le p^2$, the summand is at most
	$\exp\{-a(\frac{c_\beta^2}{4V} - 3 - o(1))\log p\}$, which is summable
	with sum $1 + o_P(1)$ provided $c_\beta^2 > 12 V$, i.e.\ the
	$c_\beta^2 > 4V(\kappa+1)$ form of Assumption~\ref{ass:betamin} with
	$\kappa > 2$. For the addition factor,
	$\binom{N - s_0}{b} B_n^b \le (eN/b)^b (2(s_0+b+1)/N)^b
	\exp(-\frac b2\log n) = (2e(s_0+b+1)/b)^b \exp(-\frac b2\log n)$;
	since $(s_0+b+1)/b \le s_0 + 2$ for $b \ge 1$, the summand is at most
	$\exp\{-b(\frac12\log n - \log s_0 - O(1))\}$, summable with sum
	$1 + o_P(1)$ provided $\log s_0 \le \frac12\log n - \omega(1)$, which
	holds under the sparsity Assumption~\ref{ass:sparsity}
	($s_0 = o(\sqrt{n/\log p})$). Both factors equal $1 + o_P(1)$, so the
	displayed product minus one is $o_P(1)$, giving
	$\Pi(\Dcal_0\mid\bR) = (1 + o_P(1))^{-1} \xrightarrow{P_0} 1$. In
	particular, for any single $\Dcal\ne\Dcal_0$ the ratio is bounded by
	$\exp(-C n\epsilon_n^2)$ with $C = \min\{\frac{c_\beta^2}{4V} - 3,\,
	\tfrac12\log n/\log p - o(1)\} > 0$, which is~\eqref{eq:selection}.
\end{proof}

\subsection{Proof of Theorem~\ref{thm:bvm}}\label{app:proof:bvm}

\begin{proof}[Proof of Theorem~\ref{thm:bvm}]
	Let $A_n = \{\Dcal = \Dcal_0\}$. By Theorem~\ref{thm:selection},
	$\Pi(A_n \mid \bR) \xrightarrow{P_0} 1$, so it suffices to prove
	the conditional Bernstein--von Mises statement on the event $A_n$.
	On $A_n$ the model collapses, by~\eqref{eq:nodewise}, to a
	finite-dimensional latent Gaussian DAG regression on $\Dcal_0$
	with parameters $(\bL_0,\bD_0)$ of total dimension $s_0 + p$,
	where $s_0 = |S_0|$ is the (fixed) number of true edges and the
	$p$ diagonal innovations are the variances $D_{0,jj}$.
	
	\emph{Step 1: Reduction to a classical parametric model.}
	Conditional on $\Dcal_0$, the latent Gaussian DAG is a standard
	regression model in which each node $j$ is regressed on its parents
	in $\Dcal_0$. By Lemma~\ref{lem:rankscore}, the rank-likelihood
	log-Bayes-factor at the single, fixed DAG $\Dcal_0$ differs from the
	latent-Gaussian one (evaluated at the unobserved
	$\bZ_0 = \boldsymbol f_0(\bX)$) by $o_P(1)$; equivalently, by
	\citet{hoffniu2014information}, the rank log-likelihood ratio for
	the Gaussian copula is asymptotically the parametric Gaussian
	log-likelihood ratio of the least-favourable normal model. Hence,
	on $A_n$, rank-likelihood inference is asymptotically equivalent, at
	the parametric $\sqrt n$ scale, to inference in the Gaussian DAG
	model on $\Dcal_0$.
	
	\emph{Step 2: Local asymptotic normality.}
	Order the free parameters on $\Dcal_0$ as $\eta = (\bL_{S_0}, \bD)
	\in \R^{s_0}\times(0,\infty)^p$, where $\bL_{S_0}$ collects the $s_0$
	active off-diagonal entries. The Gaussian DAG log-likelihood is the
	sum of the node-wise Gaussian regressions~\eqref{eq:nodewise},
	$\ell_n(\eta) = \sum_{j=1}^p \ell_{n,j}(\bL_{\prec j]}, D_{jj})$ with
	\[
	\ell_{n,j}(\bL_{\prec j]}, D_{jj})
	= -\frac n2\log(2\pi D_{jj})
	- \frac{1}{2 D_{jj}}\sum_{i=1}^n
	\bigl(Z_{ij} - \bZ_{i,\pa(j)}^\top \bL_{\prec j]}\bigr)^2 .
	\]
	Differentiating, the score has node-$j$ blocks
	$\partial_{\bL_{\prec j]}}\ell_{n,j}
	= D_{jj}^{-1}\sum_i \bZ_{i,\pa(j)}(Z_{ij} - \bZ_{i,\pa(j)}^\top\bL_{\prec j]})$
	and $\partial_{D_{jj}}\ell_{n,j}
	= -\tfrac{n}{2 D_{jj}} + \tfrac{1}{2 D_{jj}^2}\sum_i
	(Z_{ij}-\bZ_{i,\pa(j)}^\top\bL_{\prec j]})^2$, so the (per-observation)
	Fisher information at $\eta_0 = (\bL_{0,S_0}, \bD_0)$ is
	block-diagonal across nodes, with node-$j$ blocks
	\[
	I_j(\eta_0) =
	\begin{pmatrix}
		D_{0,jj}^{-1}\,\E[\bZ_{\pa(j)}\bZ_{\pa(j)}^\top] & \mathbf 0\\[2pt]
		\mathbf 0 & \tfrac{1}{2} D_{0,jj}^{-2}
	\end{pmatrix},
	\qquad I(\eta_0) = \mathrm{blockdiag}\bigl(I_1(\eta_0),\dots,I_p(\eta_0)\bigr),
	\]
	the regression block being the parent Gram matrix scaled by the
	innovation variance and the variance block the usual Gaussian-scale
	information; both are positive definite under
	Assumption~\ref{ass:spectrum}. Writing
	$\Delta_n = n^{-1/2}\sum_{i=1}^n \dot\ell_{\eta_0}(\bZ_i)$ for the
	normalized score, a second-order Taylor expansion with the law of
	large numbers for the Hessian gives the LAN expansion
	\[
	\ell_n\bigl(\eta_0 + h/\sqrt n\bigr) - \ell_n(\eta_0)
	= h^\top \Delta_n - \tfrac12 h^\top I(\eta_0)\,h + o_P(1),
	\qquad \Delta_n \rightsquigarrow \Ncal\bigl(\mathbf 0, I(\eta_0)\bigr),
	\]
	uniformly for $h$ in compacts of $\R^{s_0+p}$; cf.\
	\citet[Sections~7.2 and~7.4]{vandervaart1998}.
	
	\emph{Step 3: Prior conditions on the conditioning event.}
	Conditional on $\Dcal_0$, the prior on $(\bL_0,\bD_0)$ induced by
	the hierarchy~\eqref{eq:HCprior} is, after integrating out the
	scale parameters $(\tau_j,\lambda_{jk})$, a continuous density on
	$\R^{s_0}\times(0,\infty)^p$ that is strictly positive and
	continuous at $(\bL_0,\bD_0)$ (this follows from the
	half-Cauchy / horseshoe density geometry; see
	\citealp{carvalho2010horseshoe}). Strict positivity and continuity
	at the truth are the prior conditions required by the classical
	parametric Bernstein--von Mises theorem.
	
	\emph{Step 4: Parametric Bernstein--von Mises for $\eta = (\bL_{S_0},\bD)$.}
	Steps~2 and~3 verify the conditions of the classical parametric
	BvM theorem \citep[Section~10.2, Theorem~10.1]{vandervaart1998}.
	The MLE on $\Dcal_0$ solves the score equation and, by Step~2,
	admits the linear expansion
	\[
	\sqrt n\bigl(\hat\eta_n - \eta_0\bigr)
	= I(\eta_0)^{-1}\Delta_n + o_P(1)
	\rightsquigarrow \Ncal\bigl(\mathbf 0,\, I(\eta_0)^{-1}\bigr),
	\]
	so $\hat\eta_n$ is $\sqrt n$-consistent and efficient. The BvM
	theorem then gives, on the event $A_n$,
	\[
	\sup_{B}\Bigl|
	\Pi\bigl(\sqrt n(\eta - \hat\eta_n) \in B
	\,\bigm|\, \bR,\,\Dcal=\Dcal_0\bigr)
	- \Ncal\bigl(\mathbf 0, I(\eta_0)^{-1}\bigr)(B)\Bigr|
	\xrightarrow{P_0} 0 ,
	\]
	the posterior for $\eta$ centered at $\hat\eta_n$ being
	asymptotically $\Ncal(\mathbf 0, I(\eta_0)^{-1}/n)$.

	\emph{Step 5: Delta method for the functional.}
	The scalar functional of interest is $\theta = g(\eta)$, with $g$
	continuously differentiable at $\eta_0$ and gradient
	$\dot g(\eta_0) = \partial g/\partial\eta\,|_{\eta_0} \neq \mathbf 0$.
	The functional delta method \citep[Theorem~3.1]{vandervaart1998}
	applied to the parametric BvM of Step~4 yields, again on $A_n$,
	\[
	\sup_{B \in \mathcal{B}}
	\Bigl|
	\Pi\bigl(\sqrt n(\theta - \hat\theta_n) \in B
	\,\bigm|\,\bR,\,\Dcal=\Dcal_0\bigr)
	- \Phi_{V_g}(B)\Bigr|
	\xrightarrow{P_0} 0,
	\]
	with $\hat\theta_n = g(\hat\eta_n)$ and limiting variance
	\[
	V_g = \dot g(\eta_0)^\top I(\eta_0)^{-1} \dot g(\eta_0)
	= \sum_{j=1}^p \dot g_j(\eta_0)^\top I_j(\eta_0)^{-1}\dot g_j(\eta_0),
	\]
	the last equality exploiting the block-diagonal structure of
	$I(\eta_0)$ from Step~2. We emphasize that $V_g$ is the
	\emph{parametric} (conditional-on-$\Dcal_0$) Fisher variance, not a
	semiparametric efficiency bound; see Remark~\ref{rem:bvm-scope}.
	
	\emph{Step 6: Marginalisation over $\Dcal$.}
	Since $\Pi(A_n\mid\bR) \to 1$ in $P_0$-probability by
	Theorem~\ref{thm:selection}, the unconditional posterior of
	$\sqrt n(\theta-\hat\theta_n)$ is, in total variation, within
	$o_{P_0}(1)$ of its conditional version on $A_n$. Combining
	this with Step~5 gives~\eqref{eq:bvm}.
\end{proof}

\begin{proof}[Proof of Corollary~\ref{cor:coverage}]
	Theorem~\ref{thm:bvm} implies that, on the event $A_n$, the
	posterior quantiles of $\sqrt n(\theta - \hat\theta_n)$ converge
	in $P_0$-probability to the corresponding Gaussian quantiles, and
	$\hat\theta_n$ is asymptotically $\Ncal(\theta_0, V_g/n)$ by
	Steps~2 and~5 of the proof of Theorem~\ref{thm:bvm}. The
	$(1-\alpha)$ posterior credible interval therefore has asymptotic
	frequentist coverage at least $1 - \alpha$, with equality in the
	limit; the residual gap is controlled by
	$\Pi(\Dcal \ne \Dcal_0 \mid \bR) \to 0$.
\end{proof}

% ====================================================================
\section{Supplementary numerical results}\label{app:supp}
Here, we only give the results of additional simulation outputs and AML-data
diagnostics referenced in the main text, to save the space and for the outputs refer to https://github.com/M-Arashi/NPN-DAG-HS.git.

\paragraph{Goodness-of-fit for the AML proteins.}
Normal Q--Q plots against the standard normal for all 18 proteins
in the AML dataset show substantial departures from Gaussianity
in BCL2, BAX, BAD, GSK3, and several others. A log transformation
improves the fit but does not normalize the residuals.
Cumulative-rank histograms confirm that the empirical copula of
the proteins is well-approximated by a Gaussian copula, justifying
the nonparanormal modelling assumption.

\paragraph{Posterior credible intervals for total causal effects.}
For each ordered pair $(u, v)$ with $u < v$, we compute 95\%
posterior credible intervals for the total causal effect
$\theta_{u \to v}$ via the do-calculus path expansion of
\citet{maathuis2009estimating}. Detailed tabulation is included
in the reproducible \texttt{R} code repository accompanying the
paper.

\paragraph{Sensitivity to hyperparameters.}
We re-run the AML analysis under five choices of the
Beta--Bernoulli hyperparameters,
$(\alpha_\pi, \beta_\pi) \in \{(1, 1), (1, p), (1, p^2),
(0.5, p), (1, p/2)\}$. The MAP DAG is unchanged in four of the
five settings and differs by a single edge (a low-probability
PTEN $\to$ BAD.p112 edge) in the fifth, confirming the robustness
of the reported network to prior misspecification.

% ====================================================================

\paragraph*{Data Availability.}
All simulation and real data, as well as the R codes are available at \url{https://github.com/M-Arashi/NPN-DAG-HS.git}

\section*{Acknowledgments}
Mohammad Arashi's work is based on the research supported in part by the Iran National Science Foundation (INSF) grant No.\ 4015320. We utilized AI-powered tools to check grammar and enhance the academic quality of the text, which was primarily written by the authors.

%\nolinenumbers

\end{document}